\documentclass{article}

\usepackage{microtype}
\usepackage{graphicx}
\usepackage{subfigure}
\usepackage{booktabs} 

\usepackage{hyperref}

\usepackage[accepted]{icml2024}

\usepackage{amsmath}
\usepackage{amssymb}
\usepackage{mathtools}
\usepackage{amsthm}

\usepackage[capitalize,noabbrev]{cleveref}

\theoremstyle{plain}
\newtheorem{theorem}{Theorem}[section]

\newtheorem{lemma}[theorem]{Lemma}
\newtheorem{corollary}[theorem]{Corollary}
\theoremstyle{definition}

\newtheorem{assumption}[theorem]{Assumption}
\theoremstyle{remark}
\newtheorem{remark}[theorem]{Remark}

\usepackage[textsize=tiny]{todonotes}

\icmltitlerunning{Orthogonal Bootstrapping}

\newenvironment{iproof}{\paragraph{Informal Proof:}}{\hfill$\square$}
\newcommand{\Variance}{\text{Var}}
\usepackage{multirow}
\usepackage{bm}
\newcommand{\bX}{\bm{X}}
\newcommand{\bx}{\bm{x}}
\newcommand{\by}{\bm{y}}
\usepackage{thm-restate}

\begin{document}

\twocolumn[
\icmltitle{Orthogonal Bootstrap: Efficient Simulation of Input Uncertainty}




\begin{icmlauthorlist}
\icmlauthor{Kaizhao Liu}{yyy}
\icmlauthor{Jose Blanchet}{comp1}
\icmlauthor{Lexing Ying}{comp}
\icmlauthor{Yiping Lu}{sch1,sch2}
\end{icmlauthorlist}

\icmlaffiliation{yyy}{Department of Mathematics, Peking University, Beijing, China}
\icmlaffiliation{comp1}{Department of Management Science and Engineering, Stanford University}
\icmlaffiliation{comp}{Department of Mathematics, Stanford University}
\icmlaffiliation{sch1}{Courant Institute of Mathematical Sciences, New York University}
\icmlaffiliation{sch2}{Department of Industrial Engineering and Management Sciences, Northwestern University}

\icmlcorrespondingauthor{Kaizhao Liu}{mrzt@stu.pku.edu.cn}
\icmlcorrespondingauthor{Yiping Lu}{yiping.lu@northwestern.edu}

\icmlkeywords{Machine Learning, ICML}

\vskip 0.3in
]



\printAffiliationsAndNotice{} 

\begin{abstract}
Bootstrap is a popular methodology for simulating input uncertainty. However, it can be computationally expensive when the number of samples is large.
We propose a new approach called \textbf{Orthogonal Bootstrap} that reduces the number of required Monte Carlo replications. We decomposes the target being simulated into two parts: the \textit{non-orthogonal part} which has a closed-form result known as Infinitesimal Jackknife and the \textit{orthogonal part} which is easier to be simulated.  We theoretically and numerically show that Orthogonal Bootstrap significantly reduces the computational cost of Bootstrap while improving empirical accuracy and maintaining the same width of the constructed interval.

\end{abstract}

\section{Introduction}
The input uncertainty problem \cite{lam2018subsampling} manifests as the propagation of statistical noise from input models, typically derived and calibrated from data, into the subsequent output analysis. This noise can significantly impact the accuracy and reliability of the resulting conclusions, necessitating robust strategies for quantification and mitigation. Bootstrap \cite{stine1985bootstrap,efron1992bootstrap,efron1994introduction} is a non-parametric method that uses random resampling with replacement to estimate this uncertainty.  In this paper, we consider using Bootstrap to estimate the mean and variance of a function with input uncertainty, which has wide application in debiasing the functional estimation \cite{quenouille1949approximate,efron1982jackknife,adams1971asymptotic,cordeiro2014introduction,etter2020operator,jiao2020bias,koltchinskii2021estimation,koltchinskii2022estimation,zhou2021high,etter2021operator,ma2022correcting}, improving worst group generalization \cite{sagawa2020investigation,nguyen2022improved} and, most influential, constructing the confidence interval \cite{tukey1958bias,stine1985bootstrap,efron1992bootstrap,efron1994introduction}. 

Despite its benefits, Bootstrap is computationally demanding. Ideally, to simulate the Bootstrap estimation perfectly,  an infinite number of Monte Carlo replications of resamples is required. In practice, only a finite number of Monte Carlo replications can be actually performed, and this introduce an error in the Bootstrap estimation (that can be characterized by its variance),  which we call the \mbox{\textit{simulation error}}. Another source of error in the Bootstrap estimation arises from the randomness in the data. The former can be controlled by increasing the number of Monte Carlo replications, while the latter can not be controlled if the samples are given. To obtain a reasonable Bootstrap estimation, the number of Monte Carlo replications required should at least ensure that the simulation error is smaller than the error arising from the data.  In most applications, the number of Monte Carlo replications required scales up with the number of data points \cite{lam2018subsampling,lam2022subsampling}, where each replication can involve expensive optimization procedures to refit models, making it extremely challenging to scale Bootstrap resampling to large datasets. This drawback also occurs in bagging \cite{wager2014confidence}.

To address these issues, we propose an alternative method called \textbf{Orthogonal Bootstrap}. We use semiparametric techniques to reduce the number of required Monte Carlo replications for stochastic simulation under input uncertainty.  Our method is inspired by recently proposed double/orthogonal machine learning \cite{foster2019orthogonal,chernozhukov2018double,chernozhukov2022automatic,chernozhukov2022locally} which utilize the influence function \cite{cook1982residuals,efron1992jackknife}
to debias parameter estimates in the presence of nuisance parameters. Suppose we want to estimate the uncertainty of the output using Bootstrap. We can regard the Bootstrap method as a two-stage estimator. In the first stage, we generate resample distributions by resampling the original data with replacement and calculate the output for each resampled distribution. In the second stage, we simply calculate the variance of the outputs from the first stage. We consider the first stage to be ``nuisance" because our primary interest is in the final uncertainty, and we do not need to know the simulation output for each resampled distribution. Motivated by this, we show in this paper that we can reduce the simulation error of Bootstrap resampling by dealing with the non-orthogonal and orthogonal parts separately. The non-orthogonal part enjoys a closed form result using the influence function \cite{rousseeuw2011robust,cook1980characterizations,koh2017understanding}, also known as Infinitesimal Jackknife (IJ) \cite{jaeckel1972infinitesimal,giordano2019swiss,giordano2019higher,lu2020uncertainty,alaa2020discriminative,abad2022approximating}. For modern machine learning, the influence function can be calculated efficiently using implicit Hessian-vector products and is much faster than retraining the model \cite{cook1980characterizations,koh2017understanding,giordano2019swiss}.  Note that when only calculating the non-orthogonal part, \emph{i.e.} using the IJ method, the final variance estimate tends to be conservative  \cite{efron1981jackknife,efron1992jackknife}. In Orthogonal Bootstrap, we further simulate the orthogonal part to correct the IJ estimation. Thus, our method provides the same result as the Bootstrap method in contrast to the biased Jackknife estimator. This enables our method to enjoy the accuracy and effectiveness of the Bootstrap method such as higher order coverage for confidence interval construction  \cite{hall1986bootstrap,efron1992jackknife,hall2013bootstrap} while enjoying similar computational cost as the IJ method. We also remark here that, interestingly, our method can also be understood as using IJ as a control variate in the original Bootstrap method.

\subsection{Related Work}
Our paper uses a similar but not identical setting as in \cite{lam2022subsampling} for the input uncertainty problem. The authors also investigated the simulation effort required by performing Bootstrap. They showed how the total required simulation effort can be reduced from an order bigger than the data size in the conventional approach to an order independent of the data size via subsampling. However, their setting involves two layer of simulation in Bootstrap, while we consider simple plug-in estimator so we consider only one layer of simulation. Interestingly, although their subsampling techniques do not involve influence functions, they leveraged influence functions when proving theoretical results.

A relevant baseline for our paper is the recently proposed Cheap Bootstrap \cite{lam2022cheap}. The method also aims to maintain desirable statistical guarantees of confidence interval coverage with minimal resampling effort as low as one Monte Carlo replication. However, Cheap Bootstrap enlarges the confidence interval $t_{B,1-\alpha/2}/z_{1-\alpha/2}>1$ times, where $t_{B,1-\alpha/2}$ is the $(1-\alpha/2)-$quantile of $t_B$, the student $t-$distribution with degree of freedom $B$ where $B$ is the number of Monte Carlo replications, and $z_{1-\alpha/2}$ is the $(1-\alpha/2)-$quantile of standard normal. When the number of Monte Carlo replications is constrained, the Cheap Bootstrap will provide a confidence interval with a huge width mean.  By contrast, our proposed Orthogonal Bootstrap provides a confidence interval with the same expected width mean as the Standard Bootstrap, even when the number of Monte Carlo replications is constrained. At the same time, the Orthogonal Bootstrap can be used beyond confidence interval construction. It can be used for all settings when the Bootstrap technique is applied, for example, bias reduction for functional estimation \cite{quenouille1949approximate,efron1982jackknife,efron1992bootstrap,jiao2020bias,koltchinskii2021estimation,ma2022correcting}.

Another relevant paper to ours is \cite{kline2012score}, which proposed a higher order approximation (up to order $O(n^{-1})$) to the Wild Bootstrap  \cite{wu1986jackknife,liu1988bootstrap} . However, \cite{kline2012score}'s methodology can only be used to construct confidence intervals for \emph{linear estimators} while orthogonal bootstrap works without assuming any structure of the output functional.  Recently, \cite{zhou2021high} proposed Higher-Order Statistical Expansion (HODSE) based on the closed-form representation of the Jackknife estimator of ideal degenerate expansion of the target statistical functional. However, these types of constructions can only be used for certain statistical models, such as smooth function-of-mean model or linear regression. In contrast, our method functions as a general procedure for Bootstrap methods that doesn't depend on either the linear structure of the estimator or the smooth function-of-mean model. We achieve this by considering the Taylor expansion over the input distribution as a control variate for the Bootstrap estimator.

\subsection{Contribution}
Our contributions are summarized as follows:
\begin{itemize}
\setlength{\itemsep}{0pt}
\setlength{\parsep}{0pt}
    \item We propose Orthogonal Bootstrap, a brand new Bootstrap method that provides the same expected simulation result as the original Standard Bootstrap but with reduced simulation effort via separately treating the non-orthogonal part and orthogonal part. 
    \item Theoretically, we show that the Orthogonal Bootstrap can provably reduce the simulation error so that the required number of Monte Carlo replications required decreased from $\Omega(n)$ to $O(1)$, under the mild assumption that the performance measure has a continuous Fr\'echet derivative under the Kernel Maximum Mean Discrepancy (MMD) distance. As far as the authors know, we are the first paper to link the Bootstrap simulation error with differentiability in Kernel MMD. 
    \item Empirically, we show that Orthogonal Bootstrap can significantly improve the result on both simulated and real datasets when the number of Monte Carlo replications is limited.
\end{itemize}

\subsection{Organization of the Paper}
We organize our paper as follows. In Section \ref{section:debiasing}, we demonstrate how Orthogonal Bootstrap can be used to reduce the bias of a statistical estimator. In Section \ref{section:ci}, we demonstrate how Orthogonal Bootstrap can be used to quantify the uncertainty of a statistical estimator. In Section \ref{section:numerical}, we present some numerical examples on both simulated and real datasets.

\subsection{Notation}
Throughout the paper, we adopt the following notation conventions \cite{lam2022subsampling}. The symbol $n$ represents the sample size. The notation $F$ refers to a distribution, $\hat{F}$ stands for the empirical distribution created through sampling from $F$, and $\hat{F}^b$ signifies the empirical distribution created through bootstrap resampling from $\hat{F}$. In particular, superscripts are used to distinguish bootstrapped distributions and subscripts to denote distinct input distributions; for instance, $\hat{F}^b_i$ denotes the empirical distribution constructed from the $b$-th bootstrap resample of the $i$-th empirical distribution $\hat{F}_i$.
For a random variable $X_{i,j}$ with double subscripts, the first subscript denotes distinct input distributions, and the second subscript denotes diverse samples drawn from the $i$-th input distribution. We use $\mathbb{E}_*$ and $\Variance_*$ to denote the expectation and variance over the bootstrap resamples from the data, conditional on the original data $\hat{F}$. That is, $\mathbb{E}_*$ and $\Variance_*$ only accounts for the randomness in simulation, \emph{i.e.} the simulation error. We use $\mathbb{E}$ and $\Variance$ to denote the expectation and variance over the data distribution $F$. We also use the standard $O_p$ notations: $\Theta_p(\cdot)$, $O_p(\cdot)$, $\Omega_p(\cdot)$ in probability statement about the data distribution $F$, to only hide constants that do not depend on $n$.

\section{Debiasing via Orthogonal Bootstrap}
\label{section:debiasing}

In this section, we study the problem of estimating function/functional values when uncertainty exists on the input value \cite{quenouille1949approximate,efron1982jackknife,efron1992bootstrap,jiao2020bias,koltchinskii2021estimation,koltchinskii2022estimation,zhou2021high,etter2020operator,etter2021operator,ma2022correcting}. 
 We first formulate the problem of simulating with input uncertainty \cite{song2014advanced,lam2018subsampling}, then describe the motivation of our Orthogonal Bootstrap method. 
 
Suppose one aim to estimate a generic performance measure $\phi(F_1,\cdots,F_m)$ depend on $m$ independent input distributions $F_1,\cdots,F_m$. We consider the setting when the input distribution is unknown and only $n_i$ i.i.d. generated data $\{X_{i,1},\cdots, X_{i,n_i}\}$ from distribution $F_i$ is available, forming the empirical distributions $\hat F_i:=\sum_{j=1}^{n_i}\delta_{X_{i,j}}(x)/n_i$. The Bootstrap \cite{efron1994introduction} methods simulates $B\in\mathbb{Z}^+$ Monte Carlo replications. In each Monte Carlo replication, resampling with replacement is performed independently and uniformly from $\{X_{i,1},\cdots,X_{i,n_i}\}$, repeated $n_i$ times. This process yields sets $\{X_{i,1}^b,\cdots,X_{i,n_i}^b\} $ for each replications $b=1,\cdots,B$. Then according to the bootstrap principle, we can estimate $\mathbb{E}\phi(\hat{F}_1,\cdots,\hat{F}_m)$ by the bootstrap mean
$$
\mathbb{E}_* \phi(\hat F_{1}^{b},\cdots,\hat F_{m}^{b})=\lim_{B\rightarrow \infty}\frac{1}{B}\sum_{i=1}^B \phi(\hat F_{1}^{b},\cdots,\hat F_{m}^{b}),
$$where we recall from the notation section that $\mathbb E_*$ is the expectation with respect to resampling conditional on the original data and $\hat F_{i}^{b}=\sum_{j=1}^{n_i}\delta_{X_{i,j}^b}(x)/n_i$ is the resampled empirical distribution.

 The simulated bootstrap mean can be used for debiasing the naive estimation of function/functional values when uncertainty exists on the input value \cite{quenouille1949approximate,efron1982jackknife,efron1992bootstrap,jiao2020bias,koltchinskii2021estimation,koltchinskii2022estimation,zhou2021high,etter2020operator,etter2021operator,ma2022correcting}. The debiasing procedure appears in many applications, including inverse a noisy elliptic system \cite{etter2020operator,etter2021operator}, optimal stopping \cite{zhou2021unbiased}, online learning \cite{chen2022debiasing}
 and stochastic optimization \cite{ma2022correcting,li2020debiasing}. \cite{ma2022correcting} showed that the estimation can be improved from the naive estimator with $\Omega(n)$ times of bootstrap simulation. Bootstrap function/functional estimates in real applications is computationally expensive, as it requires retraining a machine learning model. Therefore, bootstrapping a large number of times (proportional to the number of data points) is impractical, and Bootstrap a limited number of times can lead to high variance and bad performance \cite{lam2018subsampling}.

In this paper, we propose methods to reduce resampling effort while maintaining the expected accuracy. To achieve this, suppose $\nabla \phi$ is the von Mises derivative (defined in Section \ref{appendix: von mises} in the appendix) of the statistical functional $\phi$. Then we decompose the bootstrap simulation target into the non-orthogonal part and the orthogonal part
\begin{equation*}\tiny
    \begin{aligned}
   & \mathbb{E}_* \phi(\hat F_{1}^{b},\cdots,\hat F_{m}^{b})= \underbrace{\mathbb{E}_* \nabla \phi(\hat F_1,\cdots,\hat F_m)\left[\hat F_{1}^{b}-\hat F_1,\cdots,\hat F_{m}^{b}-\hat F_m\right]}_{\text{=0}}\\&+ \mathbb{E}_*\left(\phi(\hat F_{1}^{b},\cdots,\hat F_{m}^{b})-\nabla \phi(\hat F_1,\cdots,\hat F_m)\left[\hat F_{1}^{b}-\hat F_1,\cdots,\hat F_{m}^{b}-\hat F_m\right]\right)\\
&=\mathbb{E}_*\left(\phi(\hat F_{1}^{b},\cdots,\hat F_{m}^{b})-\underbrace{\nabla \phi(\hat F_1,\cdots,\hat F_m)\left[\hat F_{1}^{b}-\hat F_1,\cdots,\hat F_{m}^{b}-\hat F_m\right]}_{\text{Control Variate}}\right)
    \end{aligned}
\end{equation*}\newline 
To compute $\nabla \phi(\hat F_1,\cdots,\hat F_m)[\hat F_{1}^{b}-\hat F_1,\cdots,\hat F_{m}^{b}-\hat F_m]$, we use the influence function technique \cite{cook1980characterizations,koh2017understanding,giordano2019higher}. Noting that typically $\nabla \phi(\hat F_1,\cdots,\hat F_m)[\tilde{F}_{1}-\hat F_{1},\cdots,\tilde{F}_{m}-\hat F_m]$ is a linear functional in $(\tilde{F}_{1},\cdots,\tilde{F}_{m})$, and the Riesz representation theorem implies the existence of a mean zero, finite variance functions $(\mathcal{I}_1(X_1),\cdots, \mathcal{I}_m(X_m))$ such that $\nabla \phi(\hat F_1,\cdots,\hat F_m)[\tilde{F}_{1}-\hat F_{1},\cdots,\tilde{F}_{m}-\hat F_m]=\sum_{i=1}^m \mathbb{E}_{X_i\sim \tilde{F}_i}\mathcal{I}_i(X_i)$ \cite{cook1980characterizations,serfling2009approximation,van2000asymptotic}. As the influence function (the non-orthogonal part) is mean zero, \emph{i.e.} $\mathbb{E}_* \nabla \phi(\hat F_1,\cdots,\hat F_m)[\hat F_{1}^{b}-\hat F_1,\cdots,\hat F_{m}^{b}-\hat F_m]=0$, it doesn't need to be simulated. We can also understand the non-orthogonal part as a control variate \cite{asmussen2007stochastic} to reduce the simulation error of bootstrap simulation. The whole idea is summarized in Figure \ref{fig:idea}.

Applying the above idea to bias correction, we can simulate the bias 
$\phi(F_1,\cdots,F_m)-\mathbb{E}\phi(\hat{F}_1,\cdots,\hat{F}_m)$ by 
\[\phi(\hat{F}_1,\cdots,\hat{F}_m)-\frac{1}{B}\sum_{b=1}^B(\hat{\phi}^b-\sum_{i=1}^m\frac{1}{n_i}\sum_{j=1}^{n_i}\mathcal{I}_i^\phi(x_{i,j}^b))\]

The whole procedure is summarized in Algorithm \ref{alg:cap}.

\begin{algorithm}
\caption{Debiasing via Orthogonal Bootstrap}\label{alg:cap}
 \textbf{Input}: A generic performance measure $\phi(F_1,\cdots,F_m)$, i.i.d samples  $\{X_{i,1},\cdots,X_{i,n_i}\}\in\mathbb{R}^{d_i}$ of $F_i$, and influence function $\mathcal{I}_i^\phi$ of $\phi$ respect to $\hat{F}_i$  \newline
 \textbf{Output}: Estimation of {\small$\phi(F_1,\cdots,F_m)$} 
 
\begin{algorithmic}
\STATE $\hat \phi\leftarrow \phi(\hat F_1,\cdots,\hat F_m)$, where $\hat F_i=\frac{1}{n_i}\sum_{j=1}^{n_i}\delta_{X_{i,j}}$
\FOR{b=1:B}
\FOR{i=1:m}
\STATE Sample $\{x_{i,1}^b,\cdots,x_{i,n_i}^b\}$ i.i.d from $\hat F_i$
\ENDFOR
\STATE $\hat\phi^b\leftarrow\phi(\hat F_1^b,\cdots,\hat F_m^b)$, where $\hat{F}_i^b=\frac{1}{n_i}\sum_{j=1}^{n_i} \delta_{x_{i,j}^b}$
\ENDFOR
\STATE Estimate $\phi(F_1,\cdots,F_m)$ by 
\begin{equation}\label{eq:debiasestimator}
    2\hat \phi-\frac{1}{B}\sum_{b=1}^B \left(\hat\phi^b-\sum_{i=1}^m\frac{1}{n_i}\sum_{j=1}^{n_i}\mathcal{I}_i^\phi(x_{i,j}^b)\right)
\end{equation}

\end{algorithmic}
\end{algorithm}


\subsection{Provable Improvement of Orthogonal Bootstrap}

In this section, we show our Orthogonal Bootstrap method can provably reduce the required Monte Carlo replications under mild assumption. For simplicity, we assume $n_1=n_2=\cdots=n_m=n$ in our theoretical investigation. We assume our simulation functional has a continuous Fr\'echet gradient in tht under the Kernel Maximum Mean Discrepancy (MMD) distance \cite{muandet2017kernel}. For background in kernel mean embeddings, we refer to Section \ref{appendix: rkhs} in the appendix. We first rewrite the simulation functional in terms of the kernel mean embeddings, \emph{i.e.}
\begin{equation}\label{equation: rkhs embed}
     \phi(F_1,\cdots,F_m)=h(\mu_1(F_1),\cdots,\mu_m(F_m))
\end{equation}
where $\mu_i$ are kernel mean embeddings using kernel $k_i$ \cite{muandet2017kernel}, $\mathcal{H}:=\mathcal{H}_1\times\cdots\times\mathcal{H}_m$ where $\mathcal{H}_i$ is the reproducing kernel Hilbert space respect to the kernel $k_i$ and $h:\mathcal{H}\to\mathbb{R}$ is a functional on $\mathcal{H}$. Denote $\mu:=\mu_1\times\cdots\mu_m$ and $F:=F_1\times\cdots F_m$ for simplicity.

\begin{restatable}{assumption}{mainassumptionone}\label{assumption: rkhs embed}
   There exists kernel mean embeddings $\mu_i:\mathcal{F}_i\to\mathcal{H}_i$ which maps $F_i$ into $\mathcal{H}_i$ for $i=1,\cdots,m$. 
    Moreover, for all $i=1,\cdots,m$, $\mathbb{E}k_i(X_i,X_i)^4<\infty$ and $\mathbb{E}k_i(X_i,Y_i)^4<\infty$ where $X_i,Y_i$ are independent samples from $F_i$.
\end{restatable}
For a functional on $\mathcal{H}$, we say that it is of class $C^1$ if its Fr\'echet derivative exists and is continuous.

\begin{restatable}{assumption}{mainassumptiontwo}\label{assumption: lip in rkhs}
The non-constant functional $h:\mathcal{H}_1\times\cdots\times\mathcal{H}_m\to\mathbb{R}$ is of class $C^1$ and its derivative is Lipschitz in the sense that 
    $|Dh(x_1)(v)-Dh(x_2)(v)|\le L\|x_1-x_2\|_{\mathcal{H}}\|v\|_{\mathcal{H}}$. Moreover, $\|\partial_ih(\mu(F))^4\|_{\mathcal{H}_i}<\infty$.
\end{restatable}

Under the above two assumption on the performance measure \ref{equation: rkhs embed}, we have

\begin{theorem}
     Let $X_{ob}$ be the Orthogonal Bootstrap estimator defined in Equation \eqref{eq:debiasestimator} and $X_{sb}$ be the Standard Bootstrap estimator defined by $X_{sb}:=2\hat \phi-\frac{1}{B}\sum_{b=1}^B \hat\phi^b$. Under Assumption \ref{assumption: rkhs embed} and Assumption \ref{assumption: lip in rkhs}, if the number of Monte Carlo replications $B\geq Cn^\alpha$ for some absolute constant $C>0$ and $\alpha\geq 0$, then the simulation error for the Orthogonal Bootstrap estimator satisfies $\Variance_*(X_{ob})=O_p(\frac{1}{n^{2+\alpha}})$ and the simulation error for the Standard Bootstrap estimator satisfies $\Variance_*(X_{sb})=\Theta_p(\frac{1}{n^{1+\alpha}})$. 
\end{theorem}
The theorem is proved by combining Theorem \ref{thm: debias} and Theorem \ref{thm: debias sb} with Theorem \ref{thm: influence exist in rkhs} in the appendix. 
Here we provide an informal proof to illustrate our idea. 

\begin{iproof} The complexity of bootstrap simulation is related to the variance of the following random variable 
{\tiny\begin{equation*}
    \xi=\phi(\hat F_{1}^{b},\cdots,\hat F_{m}^{b})-{\nabla \phi(\hat F_1,\cdots,\hat F_m)\left[\hat F_{1}^{b}-\hat F_1,\cdots,\hat F_{m}^{b}-\hat F_m\right]}.
\end{equation*}}
The random variable is at the scale of $O_p(\hat F_{1}^{b}-\hat F_1,\cdots,\hat F_{m}^{b}-\hat F_m)^2$, \emph{i.e.} $O_p(n^{-1})$. Thus its variance is at scale $O_p(n^{-2})$.  Thus to achieve  simulation error of order $O_p(n^{-\alpha})$, one needs $O(n^{\alpha-2})$ Monte Carlo replications. 
\end{iproof}
\begin{remark}
    To debias a functional estimation, one typically needs the simulation error to be $O_p(n^{-2})$ (see Theorem 2 and its proof in \cite{ma2022correcting}). According to our theorem, Standard Bootstrap needs $\Omega(n)$ Monte Carlo replications but Orthogonal Bootstrap only needs $O(1)$ Monte Carlo replications.
\end{remark}

\section{Variance Estimation via Orthogonal Bootstrap}
\label{section:ci}

In this section, we discuss how the idea of Orthogonal Bootstrap can be used to accelerate bootstrap simulation for estimating the variance, which provides a versatile nonparametric method for constructing confidence intervals and prediction intervals without detailed model knowledge. Following the same setting of the previous section, suppose one aims to estimate a generic performance measure $\phi(F_1,\cdots,F_m)$. To quantify the uncertainty of the plug-in estimator $\phi(\hat F_1,\cdots,\hat F_m)$, one needs to know its variance, which can be simulated by $\Variance_*(\phi(\hat F_{1}^{b},\cdots,\hat F_{m}^{b}))$ according to the bootstrap principle (recall that $\Variance$ without subscript denotes the simulation variance conditioned on the original data). Similar to the previous section, we decompose the variance into the variance of the non-orthogonal part, the variance of the orthogonal part, and their cross-covariance as in Equation \eqref{eq:vardecompose} as follows. 
\begin{equation}\label{eq:vardecompose}
\tiny
    \begin{aligned}
        &\Variance_*\left(\phi(\hat F_{1}^{b},\cdots,\hat F_{m}^{b})\right)=\underbrace{\Variance_* \left(\nabla \phi(\hat F_1,\cdots,\hat F_m)\left[\hat F_{1}^{b}-\hat F_1,\cdots,\hat F_{m}^{b}-\hat F_m\right]\right)}_{\text{closed form}}\\&+\Variance_*\left(\phi(\hat F_{1}^{b},\cdots,\hat F_{m}^{b})-\nabla \phi(\hat F_1,\cdots,\hat F_m)\left[\hat F_{1}^{b}-\hat F_1,\cdots,\hat F_{m}^{b}-\hat F_m\right]\right)\\&+2\text{Cov}_*\Big(\nabla \phi(\hat F_1,\cdots,\hat F_m)\left[\hat F_{1}^{b}-\hat F_1,\cdots,\hat F_{m}^{b}-\hat F_m\right],\\&\quad\phi(\hat F_{1}^{b},\cdots,\hat F_{m}^{b})-\nabla \phi(\hat F_1,\cdots,\hat F_m)\left[\hat F_{1}^{b}-\hat F_1,\cdots,\hat F_{m}^{b}-\hat F_m\right]\Big)\\
    \end{aligned}
\end{equation}
The variance of the non-orthogonal part enjoys a closed-form representation using the influence function as follows.

\begin{lemma}\label{lem: influence var} 
{\scriptsize $\Variance_* \left(\nabla \phi(\hat F_1,\cdots,\hat F_m)\left[\hat F_{1}^{b}-\hat F_1,\cdots,\hat F_{m}^{b}-\hat F_m\right]\right)\\=
        \sum_{i=1}^m\frac{1}{n_i^2}\sum_{j=1}^{n_i}(\mathcal{I}_i^\phi (X_{i,j}))^2$}
\end{lemma}
For the proof, see Lemma \ref{lem: 1st var} in the appendix.
Based on this observation and the orthogonal and non-orthogonal decomposition, we propose Algorithm \ref{alg:varestimateob}. We calculate the non-orthogonal part (influence function) using the closed-form representation and simulate the remainder part using the Monte-Carlo method. 

Under the same assumption as debiasing, Orthogonal Bootstrap provably improves the required Monte Carlo replications when simulating the variance.

\begin{theorem} \label{theorem:variance}
For simplicity, we assume $n_1=n_2=\cdots=n_m=n$. Let $X_{ob}$ be the Orthogonal Bootstrap estimator defined in Equation \eqref{eq:varianceestimator} and $X_{sb}$ be the Standard Bootstrap estimator defined by $X_{sb}:=\frac{1}{B}\sum_{b=1}^{B}(\hat{\phi}^b-\overline{\phi})^2$. Under Assumption \ref{assumption: rkhs embed} and Assumption \ref{assumption: lip in rkhs}, if the number of Monte Carlo replications $B\geq Cn^\alpha$ for some absolute constant $C>0$ and $\alpha\geq 0$, then the simulation error for the Orthogonal Bootstrap estimator satisfies $\Variance_*(X_{ob})=O_p(\frac{1}{n^{3+\alpha}})$ and the simulation error for the Standard Bootstrap estimator satisfies $\Variance_*(X_{sb})=\Theta_p(\frac{1}{n^{2+\alpha}})$.
\end{theorem}
   
   \begin{remark}
         To achieve relative consistency in variance estimation, one typically needs the simulation error to be $O_p(n^{-3})$  (see Theorem 1 in \cite{lam2022subsampling}). According to Theorem \ref{theorem:variance}, the Standard Bootstrap needs $\Omega(n)$ Monte Carlo replications but the Orthogonal Bootstrap only needs $O(1)$ Monte Carlo replications.
   \end{remark}
  The theorem is proved by combining Theorem \ref{thm: original ci} and Theorem \ref{thm: original ci sb} with Theorem \ref{thm: influence exist in rkhs} in the appendix. 

When the sample size $n$ is small, our Orthogonal Bootstrap estimator for variance can obtain negative values sometimes. An improved Orthogonal Bootstrap estimator (Algorithm \ref{alg:ciconstructioniob} in the appendix) should be used in this case. We note here that the improved Orthogonal Bootstrap estimator enjoys the same theoretical properties (Theorem \ref{thm: improved ci} in the appendix) as the original Orthogonal Bootstrap estimator. See Section \ref{appendix:iob} in the appendix for details.

\begin{algorithm*}
\caption{Variance Estimation via Orthogonal Bootstrap}\label{alg:varestimateob}
 \textbf{Input}: A generic performance measure $\phi(F_1,\cdots,F_m)$, i.i.d samples  $\{X_{i,1},\cdots,X_{i,n_i}\}\in\mathbb{R}^{d_1}$ of $F_i$, and influence function $\mathcal{I}_i^\phi$ of $\phi$ respect to $\hat{F}_i$.  \\
 \textbf{Output}: Estimation of {$\Variance\phi(\hat F_1,\cdots,\hat F_m)$}.
\begin{algorithmic}
\STATE $\hat \phi\leftarrow \phi(\hat F_1,\cdots,\hat F_m)$, where $\hat F_i=\frac{1}{n_i}\sum_{j=1}^{n_i}\delta_{X_{i,j}}$

\FOR{b=1:B}
\FOR{i=1:m}
\STATE Sample $\{x_{i,1}^b,\cdots,x_{i,n_i}^b\}$ i.i.d from $\hat{F_i}$
\ENDFOR
\STATE $\hat\phi^b\leftarrow\phi(\hat F_1^b,\cdots,\hat F_m^b)$, where $\hat{F}_i^b=\frac{1}{n_i}\sum_{j=1}^{n_i} \delta_{x_{i,j}^b}$
\STATE Calculate $\mathcal{\hat I}^b=\sum_{i=1}^m\frac{1}{n_i}\sum_{j=1}^{n_i}\mathcal{I}_i^\phi(\tilde{x}_{i,j})$
\ENDFOR
\STATE Calculate $\overline{\phi-\mathcal{I}}\leftarrow\frac{1}{B}\sum_{b=1}^{B}(\hat\phi^b-\mathcal{\hat I}^b)$, $\overline{\mathcal{I}}\leftarrow\frac{1}{B}\sum_{b=1}^{B}\mathcal{\hat I}^b$.
\STATE Estimate $\Variance_{F_1,\cdots,F_m}\phi(\hat F_1,\cdots,\hat F_m)$ by $S^2$, where
\begin{equation}\label{eq:varianceestimator}
    \begin{aligned}
        S^2&=\sum_{i=1}^m\frac{1}{n_i^2}\sum_{j=1}^{n_i}(\mathcal{I}_i^\phi (X_{i,j}))^2+\frac{1}{B}\sum_{b=1}^{B}\left(\hat\phi^b-\mathcal{\hat I}^b-\overline{\phi-\mathcal{I}}\right)^2+\frac{2}{B}\sum_{b=1}^{B}\left(\hat\phi^b-\mathcal{\hat I}^b-\overline{\phi-\mathcal{I}}\right)\left(\mathcal{\hat I}^b-\overline{\mathcal{I}}\right).
    \end{aligned}
\end{equation}
\end{algorithmic}
\end{algorithm*}

\begin{figure*}
\vspace{-0.1in}
    \centering
    \caption{ We consider modeling the relationship between resampled distribution and simulation output as nuisance estimation in orthogonal statistical learning. In Orthogonal Bootstrap, we use linear modeling for the nuisance estimation and only focus on the simulation of the orthogonal part (\emph{i.e.} the residual of linear modeling) to reduce the simulation error. }
    \includegraphics[width=4.5in]{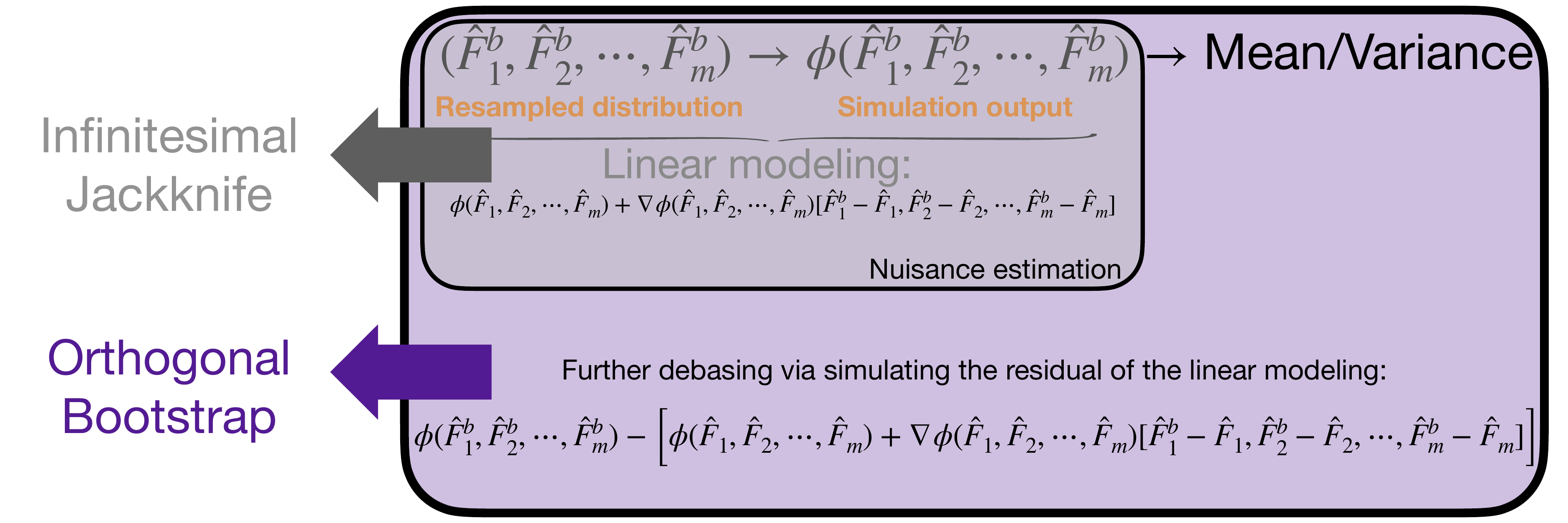}
    
    \label{fig:idea}
\vspace{-0.2in}
\end{figure*}

\begin{table*}[]
\centering
\tiny
\caption{$95\%$ confidence interval performances with different Bootstrap methods: Standard Bootstrap, Cheap Bootstrap \cite{lam2022cheap} and Orthogonal Bootstrap. We also added the Infinite Jackknife baseline here. The Standard Bootstrap is significantly under-coverage when $B=2$ to $B=5$. For example, the Standard Bootstrap method only achieved $59\%$ coverage for estimating the variance of folded normal when only 2 Monte Carlo replications are applied. In contrast, Orthogonal Bootstrap can achieve the target coverage when $B$ is small. Compared with Cheap Bootstrap \cite{lam2022cheap}, Orthogonal Bootstrap does not enlarge the length of the constructed confidence interval. Our method also outperforms the Infinitesimal Jackknife.}
\label{table: ci}

\vspace{0.1in}
\begin{tabular}{cc||cc||cc||cc||cc}
 \hline
\multirow{2}{*}{}    & \multirow{2}{*}{$B$} & \multicolumn{2}{c||}{{\textbf{Variance of Folded normal}}}                    & \multicolumn{2}{c||}{\textbf{Variance of Double exponential}}                   & \multicolumn{2}{c||}{\textbf{Correlation of Bivariate Lognormal}}                    & \multicolumn{2}{c}{\textbf{Linear Regression}}              \\ \cline{3-10} \textbf{Method}
                     &                      & \multicolumn{1}{c|}{\textbf{Coverage}} & {\textbf{Width} (st. dev.)} & \multicolumn{1}{c|}{\textbf{Coverage}} & {\textbf{Width }(st. dev.)} & \multicolumn{1}{c|}{\textbf{Coverage}} & {\textbf{Width} (st. dev.)} & \multicolumn{1}{c|}{\textbf{Coverage}} & {\textbf{Width} (st. dev.)}\\ \hline\hline
Standard Bootstrap   &       2               & \multicolumn{1}{c|}{0.591}                  &    {0.043(0.033)}                  & \multicolumn{1}{c|}{0.580}                  &     {0.302(0.238)}                    & \multicolumn{1}{c|}{0.566}                  &{0.112(0.094)}                     & \multicolumn{1}{c|}{0.632}                  & {0.361(0.260)}                     \\ \hline
Cheap Bootstrap  \cite{lam2022cheap}     &   2                   & \multicolumn{1}{c|}{0.956}                  &  {0.145(0.077)}                    & \multicolumn{1}{c|}{0.952}                  &  {1.085(0.606)}                   & \multicolumn{1}{c|}{0.936}                  &{0.385(0.239)}                      & \multicolumn{1}{c|}{0.955}                  &  {1.199(0.615)}                     \\ \hline
Orthogonal Bootstrap &   2                   & \multicolumn{1}{c|}{0.952}                  &  {0.076(0.008)}                     & \multicolumn{1}{c|}{0.949}                  &    {0.552(0.077)}                  & \multicolumn{1}{c|}{0.911}                  &{0.194(0.063)}                      & \multicolumn{1}{c|}{0.946}                  &   {0.615(0.076)}                   \\ \hline\hline
Standard Bootstrap   &       5               & \multicolumn{1}{c|}{0.838}                  &          {0.063(0.024)}            & \multicolumn{1}{c|}{0.834}                  &   {0.454(0.180)}                   & \multicolumn{1}{c|}{0.788}                  &  {0.160(0.079)}                    & \multicolumn{1}{c|}{0.840}                  &  {0.533(0.191)}                    \\ \hline
Cheap Bootstrap \cite{lam2022cheap}      &   5                   & \multicolumn{1}{c|}{0.946}                  &   {0.096(0.032)}                   & \multicolumn{1}{c|}{0.945}                  &    {0.674(0.240)}                 & \multicolumn{1}{c|}{0.918}                  &{0.245(0.111)}                      & \multicolumn{1}{c|}{0.954}                  &  {0.779(0.250)}                    \\ \hline
Orthogonal Bootstrap &   5                   & \multicolumn{1}{c|}{0.954}                  &          {0.076(0.007)}             & \multicolumn{1}{c|}{0.950}                  &           {0.549(0.076)}          & \multicolumn{1}{c|}{0.913}                  &{0.195(0.066)}                      & \multicolumn{1}{c|}{0.950}                  &      {0.623(0.039)}                \\ \hline\hline
Standard Bootstrap   &       10               & \multicolumn{1}{c|}{0.906}                  &   {0.069(0.018)}                  & \multicolumn{1}{c|}{0.896}                  &    {0.510(0.143)}                  & \multicolumn{1}{c|}{0.855}                  &{0.178(0.073)}                      & \multicolumn{1}{c|}{0.905}                  &       {0.576(0.144)}               \\ \hline
Cheap Bootstrap \cite{lam2022cheap}      &   10                   & \multicolumn{1}{c|}{0.951}                  &    {0.084(0.021)}                  & \multicolumn{1}{c|}{0.950}                  &    {0.610(0.160)}                  & \multicolumn{1}{c|}{0.933}                  &{0.215(0.082)}                      & \multicolumn{1}{c|}{0.958}                  &    {0.693(0.156)}                  \\ \hline
Orthogonal Bootstrap &   10                   & \multicolumn{1}{c|}{0.950}                  &    {0.076(0.007)}                 & \multicolumn{1}{c|}{0.955}                  &       {0.548(0.074)}               & \multicolumn{1}{c|}{0.929}                  &  {0.189(0.058)}                    & \multicolumn{1}{c|}{0.954}                  &   {0.624(0.028)}                     \\ \hline\hline
Infinitesimal Jackknife &                      & \multicolumn{1}{c|}{0.937}                  &   {0.076(0.008)}                  & \multicolumn{1}{c|}{0.931}                  &       {0.548(0.075)}               & \multicolumn{1}{c|}{0.899}                  &{0.191(0.058)}                      & \multicolumn{1}{c|}{0.942}                  &   {1.679(0.150)}                 \\ \hline

\end{tabular}
\end{table*}


\section{Numerical Examples}
\label{section:numerical}

In this section, we test the numerical performances of our Orthogonal Bootstrap and compare it with the Standard Bootstrap. For confidence and prediction interval examples, we also compared our method with the recently proposed Cheap Bootstrap method \cite{lam2022cheap}. We demonstrated that our Orthogonal Bootstrap achieved significantly smaller bias and higher empirical coverage probability over the original Bootstrap method when the number of Monte Carlo replications is small. Although Cheap Bootstrap \cite{lam2022cheap} can also provide comparable empirical coverage probability to our method, the mean width of the interval constructed by Cheap Bootstrap is much longer.  Our Orthogonal Bootstrap method provides interval with higher empirical coverage probability but {achieves the same expected width as the original Bootstrap.} The experiment details are left in Section \ref{appendix:experiment} in the appendix.

\subsection{Debiasing} 
\label{subsection:biascorrection}

We consider four numerical examples following \cite{ma2022correcting}. For all of the examples, we compute the the average bias (BIAS) of the estimates across 1000 experiments. We run the naive estimator without debiasing, Standard Bootstrap, and our Orthogonal Bootstrap for a small number of Monte Carlo replications $B=2,3,4,5,6,7,8,9,10$. To ensure robustness and reliability, we repeat the bootstrap procedure ten times and provide insights through the reporting of quantiles at the 5th, 50th, and 95th percentiles. The results is shown in Figure \ref{fig:BIAS}.

\paragraph{Function of Mean} We simulate the function-of-mean model, namely estimating $\phi=g(\boldsymbol{\mu})$ where $\boldsymbol{\mu}=\mathbb{E}\mathbf{X}$ for a $d$-dimensional random vector $\mathbf{X}$ and $g:\mathbb{R}^d\rightarrow\mathbb{R}$ is a function. We have i.i.d random sample $\{\mathbf{X}_1,\cdots,\mathbf{X}_n\}$ of random vector $\mathbf{X}$. 
We consider an ellipsoidal estimation problem $g(\boldsymbol{\mu})=\|\boldsymbol{\mu}\|_2^2$ and a fourth-order polynomial estimation problem $g(\boldsymbol{\mu})=\|\boldsymbol{\mu}\|_2^4$. The underlying distribution is set to be $\mathcal{N}(0.2\mathbf{1_d},I_{d})$. We use $d=25$ and a sample size $n=100$. The influence function of $g(\boldsymbol{\mu})=\|\boldsymbol{\mu}\|^2$ is $\mathcal{I}^g(\mathbf{x})=2(\mathbf{x}^T\hat{\boldsymbol{\mu}}-\|\hat{\boldsymbol{\mu}}\|^2)$, and the influence of $g(\boldsymbol{\mu})=\|\boldsymbol{\mu}\|^4$ follows by the chain rule.

\paragraph{Entropy}
We consider estimating entropy for discrete probability distributions, namely estimating $\phi=-\sum_{i=1}^d p_i\ln p_i$ where $p=(p_1,\cdots,p_d)$ satisfies $p_i>0$ and $\sum_{i=1}^d p_i =1$. We generate the groundtruth distribution $p$ from symmetric Dirichlet distribution with parameter $\alpha_i=1$ for all $i=1,\cdots,d$. Noisy observations of $p$ are the empirical distributions of single samples from $p$. We use a sample size $n=1000$ and dimension $d=100$. The influence function is $\mathcal{I}^\phi(x)= \sum_{i=1}^d (x_i-\hat{p}_i)(\log \hat{p}_i +1)$.

\paragraph{Constrained Optimization Problem}
The debiasing method can also be applied to optimization problems with randomness. Here we consider estimating $\phi = \arg\min_{\mathbf{x}} \mathbf{x}^T B \mathbf{x}$ under the constrain $A\mathbf{x}=\mathbf{b}$, where $\mathbf{x}\in\mathbb{R}^p$, $B\in\mathbb{R}^{p\times p}$ is a positive definite matrix, $A \in \mathbb{R}^{d\times p}$, and $\mathbf{b} \in \mathbb{R}^d$ while we only have access to its noisy observations, i.e. $\mathbf{b}\sim \mathcal{N}(\mathbf{b}^*,I_{d\times d})$. We sample $\mathbf{b}^*$ uniformly from the unit sphere $\mathbb{S}^{d-1}$. We use a sample size $n=10$, dimension $d=100$ and $p=200$.

To compute the influence function, we do not solve the quadratic programming problem directly and differentiate. Instead, we utilize a general procedure that can be applied to optimization problems lacking an explicit expression for their solutions by differentiating the KKT conditions of the constrained optimization problem. The details can be found Section \ref{appendix: constrained optimization} in the appendix.

\begin{figure}
    \centering
    \includegraphics[width=1.5in]{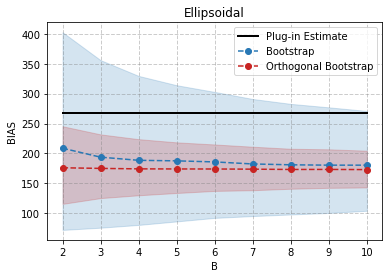}
    \includegraphics[width=1.5in]{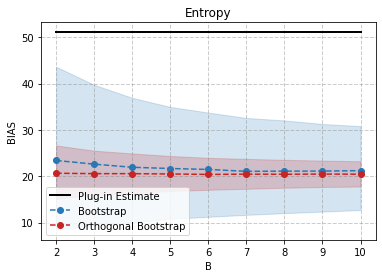}
    \includegraphics[width=1.5in]{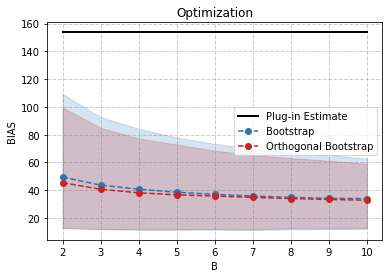}
    \includegraphics[width=1.5in]{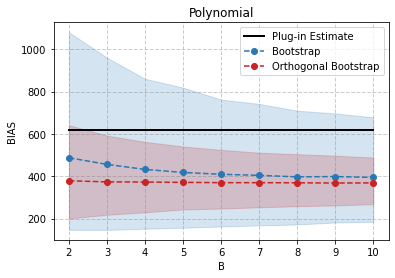}
    \caption{Orthogonal Bootstrap can significantly reduce the simulation output for the examples shown in \cite{ma2022correcting} when the number of Monte Carlo replications is limited. The $x$-axis represents the number of Monte Carlo replications and $y$-axis denotes the bias produced by the estimation. The shaded area represents the 90\% quantile interval for repeated simulations. Orthogonal Bootstrap can significantly reduce the simulation error.}
    \label{fig:BIAS}
\end{figure}

\begin{table*}[htbp]
\small
\centering
\caption{$95\%$ prediction interval performances with different Bootstrap methods: Standard Bootstrap, Cheap Bootstrap \cite{lam2022cheap}, and Orthogonal Bootstrap. Results are averaged over 3 random seeds. Our method can achieve 95\% coverage with the minimum times of Bootstrap without enlarging the prediction interval length.}
\label{realdata}
\vspace{0.1in}
\begin{tabular}{cc||c|c||c|c||c|c}
\hline
\multirow{2}{*}{} & \multirow{2}{*}{{$B$}} & \multicolumn{2}{c||}{{\textbf{Yacht Hydrodynamics}}} & \multicolumn{2}{c||}{{\textbf{Energy Efficiency}}} & \multicolumn{2}{c}{{\textbf{Kin8nm}}} \\
\cline{3-8}
\cline{3-8}
\textbf{Method} & & {\textbf{Coverage}} & {\textbf{Width}} & {\textbf{Coverage}} & {\textbf{Width}} & {\textbf{Coverage}} & {\textbf{Width}} \\
\hline\hline
Standard Bootstrap & 2 & 0.86 & 37.02 & 0.932 & 16.29 & 0.9455 & 0.6216 \\
\hline
Cheap Bootstrap \cite{lam2022cheap} & 2 & 1.00 & 82.26 & 1.00 & 36.03 & 1.00 & 1.3746 \\
\hline
Orthogonal Bootstrap & 2 & 0.99 & 45.88 & 0.951 & 15.54 & 0.9463 & 0.6224 \\
\hline\hline
Standard Bootstrap & 5 & 0.86 & 35.21 & 0.962 & 17.27 & 0.9455 & 0.6154 \\
\hline
Cheap Bootstrap \cite{lam2022cheap} & 5 & 0.86 & 46.42 & 0.993 & 22.77 & 0.9772 & 0.8110 \\
\hline
Orthogonal Bootstrap & 5 & 0.97 & 42.08 & 0.962 & 17.40 & 0.9455 & 0.6158 \\
\hline
\end{tabular}
\end{table*}

\subsection{Confidence Interval Construction} 
\label{subsection:CIconstructionexample}
In this section, we aim to construct the confidence interval \cite{hall1986bootstrap,hall1988bootstrap,efron1992bootstrap} with minimal resampling effort using our Orthogonal Bootstrap technique. 
The confidence interval constructed by the Standard Bootstrap method is
\begin{equation*}
\tiny
    \begin{aligned}
        [\hat\phi-z_{1-\alpha/2}\sqrt{\Variance_*\left(\phi(\hat F_{1}^{b},\cdots,\hat F_{m}^{b})\right)}, \hat\phi+z_{1-\alpha/2}\sqrt{\Variance_*\left(\phi(\hat F_{1}^{b},\cdots,\hat F_{m}^{b})\right)}],
    \end{aligned}
\end{equation*}

where $z_{1-\alpha/2}$ being the $(1-\alpha/2)-$quantile of the standard normal, $\hat\phi=\phi(\hat F_1,\cdots,\hat F_m)$ is the plug-in estimator of $\phi(F_1,\cdots,F_m)$ and $\hat F_i=\frac{1}{n}\sum_{j=1}^n \delta_{X_{i,j}}$. We use Orthogonal Bootstrap to estimate $\Variance_*(\phi(\hat F_{1}^{b},\cdots,\hat F_{m}^{b}))$. The details of the algorithm for confidence interval construction is provided in Algorithm \ref{alg:ciconstructionob} in the appendix.

\paragraph{Elementary Examples} Following \cite{lam2022cheap}, we consider a folded standard normal (\emph{i.e.},$|N(0,1)|$) and double exponential with rate 1 (\emph{i.e.} $\text{Sgn}\times \text{Exp}(1)$, where $\text{Sgn}=+1$ or $-1$ with equal probability and is independent with $\text{Exp}(1)$). We aim to provide confidence intervals for their variance. The two setups have explicit ground truth,   namely $1-2/\pi$ and $2$, respectively. 
The influence function of variance is $\mathcal{I}(x)=(x-\hat \mu)^2-\hat \sigma^2$, where $\hat \mu$ is the empirical mean and $\hat\sigma^2$ is the empirical variance. 
The third example, also follows \cite{lam2022cheap}, is estimating the correlation of bivariate lognormal (\emph{i.e.}$(e^{Z_1},e^{Z_2})$, where $(Z_1,Z_2)$ is the bivariate normal with mean zero, unit variance and correlation 0.5). The ground truth of this example is also known as $(e^{3/2}-e)/(e^2-e)$.
The influence function of correlation can be determined by the influence of covariance and the chain rules.
We use a sample size $n=1000$ for all three examples.  We run  1000 independent simulations and report the empirical coverage and the mean width of the confidence interval. We run the Standard Bootstrap, Cheap Bootstrap \cite{lam2022cheap} and Orthogonal Bootstrap using a small number of Monte Carlo replications $B=2,5,10$. We set the confidence level $(1-\alpha)$ as $95\%$. For each setting, we repeat our experiments 1000 times and report the empirical coverage, interval width mean and standard deviation. The result is shown in Table \ref{table: ci}.  

\paragraph{Regression Problem} We apply our method to a linear regression problem in \cite{sengupta2016subsampled,lam2022cheap}. We aim to fit a model $Y=\beta_1X_1+\cdots+\beta_dX_d+\epsilon$ where we set dimension $d=50$  and use data $\{(x_{1,i},\cdots,x_{d,i},Y_i)\}_{i=1}^n$ of size $n=5000$ to fit the model. We set $\log(X_i)\sim N(0,1)$ and $\epsilon\sim 10*N(0,1)$ as the data generating process. The influence function is easy to calculate in this case as the estimator is a $M$-estimator. For example, \cite{cook1980characterizations,koh2017understanding} calculate the influence function for  M-estimator. Specifically, consider the influence function of the M-estimator $\hat \theta = \arg\min_\theta \frac{1}{n}\sum_{i=1}^n L(z_i,\theta)$ , where $z_i$ represents training data and labels, $L(z,\theta)$ is the loss function for data $z$ and parameter $\theta$. The influence function for the parameter $\theta$ at point $z$ is  
\begin{equation}
    \begin{aligned}
        \mathcal{I}^{\theta}(z)&=\frac{d\arg\min_\theta \frac{1}{n}\sum_{i=1}^n L(z_i,\theta)+\epsilon L(z,\theta)}{d\epsilon}\bigg|_{\epsilon=0}\\
        &=-H_{\hat\theta}^{-1}\nabla_\theta L(z,\hat \theta),
    \end{aligned}
\end{equation}
where $H_{\hat \theta}=\frac{1}{n}\sum_{i=1}^n \nabla_\theta^2 L(z_i,\hat\theta)$ is the empirical hessian of the loss function and is positive definite by assumption. As shown in \cite{koh2017understanding,alaa2020discriminative}, the influence function can be solved efficiently via inverse hessian vector product combined with the modern  auto-grad systems
like TensorFlow \cite{abadi2016tensorflow} and Pytorch \cite{paszke2019pytorch}. The calculation of the influence function is much faster than retraining the model.  In this example, we run the Standard Bootstrap, Cheap Bootstrap \cite{lam2022cheap}, and our Orthogonal Bootstrap using a small number of Monte Carlo replications $B=2,5,10$ and report the empirical coverage, interval width mean, and standard deviation for the first coefficient in Table \ref{table: ci}.

\subsection{Prediction Interval Construction for Real Data}
In this section, we use Orthogonal Bootstrap to accelerate the pivot Bootstrap method \cite{contarino2022constructing} to construct prediction intervals for neural networks. Following \cite{alaa2020discriminative}, we conduct our experiments on 3 UCI benchmark datasets for regression: yacht hydrodynamics, energy efficiency \cite{Dua:2019} and kin8nm. 
We partition these datasets into distinct training and testing subsets. For the training datasets, we employ two-layer neural networks with a hidden dimension of 100 and a hyperbolic tangent (tanh) activation function. Subsequently, we apply Standard Bootstrap, Cheap Bootstrap, and our proposed Orthogonal Bootstrap to simulate the variance of the resampling step in the pivot Bootstrap method \cite{contarino2022constructing}. 
The construction details for prediction intervals can be found in Algorithm \ref{alg:piconstructionob} in the appendix.
 and the training procedures specific to each dataset 
We set the target coverage to be $(1-\alpha)=0.95$. We report the empirical coverage and interval width mean in Table \ref{realdata}, where the empirical coverage is the percentage of test data points which fall into the corresponding prediction interval. Our reported results are averaged over three random seeds to ensure robustness and reliability.

\section{Conclusion and Discussion}
In summary, our paper introduces the concept of Orthogonal Bootstrap, a novel technique that streamlines the process of bootstrap resampling. By separately treating the non-orthogonal (influence function) and orthogonal parts, akin to utilizing  Infinitesimal Jackknife estimator as a control variate, we effectively reduce the number of required Monte Carlo replications. This innovation allows our method to maintain the higher-order coverage properties associated with traditional Bootstrap methods while simultaneously decreasing the computational burden. It's important to note that the control variate significantly reduces the computational costs involved in the simulation process while it doesn't alter the expected length of the constructed confidence interval. In essence, Orthogonal Bootstrap presents a practical and efficient solution for improving the accuracy and speed of Bootstrap resampling, making it a valuable tool for statistical analysis and inference.

\section{Impact Session}

Our paper contributes to fast and accurate quantification of the uncertainty for larger scale machine learning problems, which giving policy maker an idea of reliability of the AI prediction and is vital for risk management and help to be responsible in contexts where AI decisions have significant ethical implications. The theoretical results presented in the paper will have no ethical impact. 



\bibliography{example_paper}
\bibliographystyle{icml2024}

\newpage
\appendix
\onecolumn
\section*{Organization of the Appendix}
The appendix is structured as follows. In Section \ref{appendix: expansion}, we provide some background on the von Mises expansion of the statistical functional \cite{serfling2009approximation}, kernel mean embeddings \cite{muandet2017kernel}, and multivariate calculus on Banach spaces \cite{chang2005methods}. Similar to the development of \cite{lam2022subsampling}, we propose general assumptions that theoretically guarantee the success of our method in Section \ref{appendix: theory}. We subsequently offer a range of illustrative examples where this assumption holds in the setting of kernel mean embeddings. Finally we present the details of our experiment and several additional experiment results in Section \ref{appendix:experiment}.

\section{Preliminaries}\label{appendix: expansion}
\subsection{The Von Mises Expansion}\label{appendix: von mises}
In this section, we present the Von Mises expansion, a distributional analog of the Taylor expansion applied for a statistical functional $T$. Given two points $F$ and $G$ in a collection $\mathcal{F}$ of distributions, we consider the taylor expansion of the statistic function $T$ over the line segment in $\mathcal{F}$  joining $F$ and $G$  consists of the set of distribution functions $\{(1-\lambda)F+\lambda G,0\le \lambda\le 1\}$, \emph{i.e.}
\begin{equation}
    \begin{aligned}
    \label{eq:taylor}
        T(G)-T(F)&=d_1 T(F;G-F)+\frac{1}{2!}d_2T(F;G-F)+\cdots,\\
        &=\sum_{k=1}^m \frac{1}{k!}d_kT(F;G-F)+\frac{1}{(m+1)!}\frac{d^{m+1}}{d\lambda^{(m+1)}}T(F+\lambda(G-F))|_{\lambda^\ast} \end{aligned}
\end{equation}
where $d_k(F;G-F)$ is the $k$-th order von Mises diffferential of $T$ at $F$ in the direction of $G$ to be
$$
d_k T(F;G-F)=\frac{d^k}{d\lambda^k} T(F+\lambda(G-F))|_{\lambda=0+}.
$$
Note that the von Mises differential is defined in Gateaux's manner (see Definition 1.1.2 in \cite{chang2005methods}). In typical cases, the  $k$-th order von Moses differential $d_k T(F;G-F)$ is always $k$-linear \cite{fernholz2012mises,abad2022approximating}, \emph{i.e.} there exists a function $T_k[x_1,\cdots,x_k],(x_1,\cdots,x_k)\in\mathbb{R}^k$ such that
$$
d_k T(F;G-F)=\int\cdots \int T_k[F;x_1,\cdots,x_k]\prod_{i=1}^k d[G(x_i)-F(x_i)]
$$
holds for all $G$ \cite{dayal1977converse} (e.g. Lemma \ref{lemma: Gateaux and Frechet}). Following Lemma 6.3.2.A/B in \cite{serfling2009approximation}, we use the V-statistics representation of $d_k T(F;F_n-F)$, we can have the following von Mises expansion \cite{mises1947asymptotic,filippova1962mises,reeds1976definition,van2000asymptotic,serfling2009approximation}
$$
\theta(G)=\theta(F)+\mathbb{E}_G\phi_1(X)+\frac{1}{2}\mathbb{E}_G \phi_2(X_1,X_2)+\cdots=\theta(F)+\sum_{k=1}^\infty \frac{1}{k!} \mathbb{E}_G\phi_k(X_1,\cdots,X_k)
$$
The function $\phi_1(x)$ is known as the influence function of $\theta$ and similarly $\phi_k(x_1,\cdots,x_k)$ is the $k$-th order influence function defined as
$$
\phi_k(x_1,\cdots,x_k)=\frac{d}{ds_1}|_{s_1=0}\cdots \frac{d}{ds_k}|_{s_k=0} \theta((1-\sum s_i)F+\sum s_i\delta_{x_i})
$$
Now we document some lemma that are useful for our theoretical development. We consider a statistical functional resembling the form of the $k$-th order von Mises differential.
\begin{lemma}[Section 6.3.2 Lemma A, \cite{serfling2009approximation}]\label{lemma:demean}
    Let $F$ be fixed and $h(x_1,\cdots,x_m)$ be given. A functional of the form 
    \[T(G)=\int\cdots\int h(x_1,\cdots,x_m)\prod_{i=1}^m\mathrm{d}[G(x_i)-F(x_i)]\]
    can be written as a functional of the form 
    \[T(G)=\int\cdots\int \tilde h(x_1,\cdots,x_m)\prod_{i=1}^m\mathrm{d}[G(x_i)]\]
    where the definition of $\tilde h$ depends on $F$. Moreover, we can have $\int \tilde h(x_1,\cdots,x_m)\mathrm{d}F(x_i)=0$ for all $i\in [m]$.
\end{lemma}
\begin{proof}
    Take 
    \begin{align*}
        \tilde h(x_1,\cdots,x_m)=&h(x_1,\cdots,x_m)-\sum_{i=1}^m\int h(x_1,\cdots,x_m)\mathrm{d}F(x_i)\\
        &+\sum_{i<j}\int\int h(x_1,\cdots,x_m)\mathrm{d}F(x_i)\mathrm{d}F(x_j)-\cdots\\
        &+(-1)^m\int\cdots\int h(x_1,\cdots,x_m)\prod_{i=1}^m\mathrm{d}F(x_i),
    \end{align*}
    then $\tilde h$ satisfies all properties claimed in the lemma.
\end{proof}
Using this alternative representation, we can obtain two lemma which gives the rate of the moment of the $V$-statistics type of functional in Lemma \ref{lemma:demean}.
\begin{lemma}[Section 6.3.2 Lemma B, \cite{serfling2009approximation}]\label{lemma:variance} Suppose that $\mathbb{E}_F\{h(X_1,\cdots,X_{i_m})^2\}<\infty$ for all $1\le i_1,\cdots,i_m\le m$. Then
\begin{equation}
    \begin{aligned}
    \label{secondorder}
        \mathbb{E}_F\left\{\left(\int\cdots\int h(x_1,\cdots,x_m)\prod_{i=1}^m d[\hat{F}(x_i)-F(x_i)]\right)^2\right\}=O(n^{-m}),
    \end{aligned}
\end{equation}
where $\hat{F}$ is the empirical distribution for i.i.d observations $X_1,X_2,\cdots,X_n$ of distribution function $F$. 
\end{lemma}

We replicate the proof of Lemma \ref{lemma:variance} here, for we use the same proof technique to prove Lemma \ref{lemma:four}
\begin{proof} 
Let $\tilde h$ be defined as in Lemma \ref{lemma:demean}. The left-hand side of (\ref{secondorder}) is given by 
$$
n^{-2m}\sum_{i_1=1}^n\cdots\sum_{i_m=1}^n\sum_{j_1=1}^n\cdots\sum_{j_m=1}^n\mathbb{E}_F \tilde h(X_{i_1},\cdots,X_{i_m})\tilde h(X_{j_1},\cdots,X_{j_m}).
$$

For $\int \tilde(X_1,\cdots,X_k)dF(x_i)=0$ for $1\le i\le k$, thus the typical term in the upper part may be possibly nonzero only if the sequence of indices $i_1,\cdots,i_m,j_1,\cdots,j_m$ contains each member at least twice. The number of such terms is clearly $O(n^m)$, and by the assumption that $\mathbb{E}_F\{h(X_1,\cdots,X_{i_m})^2\}<\infty$, we know (\ref{secondorder}) holds. 
\end{proof}

\begin{lemma}\label{lemma:four} Suppose that $\mathbb{E}_F\{h(X_1,\cdots,X_{i_m})^4\}<\infty$ for all $1\le i_1,\cdots,i_m\le m$. Then
\begin{equation}
    \begin{aligned}
    \label{fourthorder}
        \mathbb{E}_F\left\{\left(\int\cdots\int h(x_1,\cdots,x_m)\prod_{i=1}^m d[\hat{F}(x_i)-F(x_i)]\right)^4\right\}=O(n^{-2m}),
    \end{aligned}
\end{equation}
where $\hat{F}$ is the empirical distribution for i.i.d observations $X_1,X_2,\cdots,X_m$ of distribution function $F$. 
\end{lemma}
\begin{proof} The left-hand side of (\ref{fourthorder}) is given by
    \begin{equation*}
        \begin{aligned}
n^{-4m}\sum_{i^1_1=1}^n\cdots\sum_{i^1_m=1}^n\sum_{i^2_1=1}^n\cdots\sum_{i^2_m=1}^n\sum_{i^1_3=1}^n\cdots\sum_{i^3_m=1}^n\sum_{i^4_1=1}^n\cdots\sum_{i^4_m=1}^n\mathbb{E}\prod_{j=1}^4 h(X_{i^j_1},\cdots,X_{i^j_m})
        \end{aligned}
    \end{equation*}
For $\int h(X_1,\cdots,X_k)dF(x_i)=0$ for $1\le i\le k$, thus the typical term in the upper part may be possibly nonzero only if $i_1^1,\cdots,i_m^1,i_1^2,\cdots,i_m^2,i_1^3,\cdots,i_m^3,i_1^4,\cdots,i_m^4$ contains each member at least twice. The number of such terms is clearly $O(n^{2m})$, and by the assumption that $\mathbb{E}_F\{h(X_1,\cdots,X_{i_m})^4\}<\infty$, we know (\ref{fourthorder}) holds.
\end{proof}

The Von Mises Expansion exists for many functionals, for example, divergence \cite{serfling2009approximation,kandasamy2014influence} and (regularized) M-estimation \cite{serfling2009approximation,giordano2019swiss,giordano2019higher}. For the Taylor expansion \ref{eq:taylor} to be rigorous, from Lemma \ref{lemma:variance}, we showed that it is suffices to show that $n^{m/2}\sup_{0\le\lambda\le1}\left|\frac{d^{m+1}}{d\lambda^{m+1}}T(F+\lambda(F_n-F))\right|\overset{p}{\rightarrow}0$, or to bound the remainder $n^{m/2}R_{mn}=n^{m/2}\left(T(F_n)-T(F)-\sum_{k=1}^m \frac{1}{k!}d_k(F;F_n-F)\right)\overset{p}{\rightarrow}0$.

\subsection{Expansion on Normed Space}
An alternative functional derivative can be defined via Fr\'echet derivative if the space of distribution is equipped with a norm.
Let $\mathcal{D}:=\{\Delta:\Delta=c(G-H)|c\in\mathbb{R},G\in\mathcal{F},H\in\mathcal{F}\}$ be the linear space generated by differences . Let $\mathcal{D}$ be equipped with a norm $\|\cdot\|$.
The first order Fr\'echet derivative $T(F;G-F)$ is a linear functional which satisfies 
\begin{equation}\label{def: 1st frechet derivative}
    \lim_{G\to F}\frac{|T(G)-T(F)-T(F;G-F)|}{\|G-F \|} =0
\end{equation}
if it exists.

The following lemma shows that if the first order Fr\'echet derivative exists, then the first order von Mises differential exists and is linear.
\begin{lemma} [Section 6.2.2, Lemma A, \cite{serfling2009approximation}]\label{lemma: Gateaux and Frechet}
    Suppose that $T$ has a Fr\'echet derivative at $F$ with respect to $\|\cdot\|$,
    then for any $G$, $d_1 T(F;G-F)$ exists and 
    \[d_1 T(F;G_F)=T(F;G-F).\] 
\end{lemma}
\begin{lemma} [Section 6.2.2, Lemma B, \cite{serfling2009approximation}]
    Let $T$ have a differential at $F$ with respect to $\|\cdot\|$.
    Let $\{X_i\}_{i=1}^n$ be i.i.d. sample from $F$ and $\hat{T}_n=\frac{1}{n}\sum_{i=1}^{n}\delta_{X_i}$.
    If $\sqrt{n}\|F_n-F\|=O(1)$, then $\sqrt{n}(T(\hat{F})-T(F)-T(F;\hat{F}-F))\to^p 0$.
\end{lemma}
We can define higher order Fr\'echet derivatives in a similar way \cite{chang2005methods}. 
The second order derivative is defined to be a bilinear mapping which satisfies
\begin{equation}\label{def: 2nd frechet derivative}
    \|T(G)-T(F)-T(F;G-F)-\frac{1}{2}T(F;G-F,G-F)\|=o(\|G-F\|^2)
\end{equation}
if it exists, and the $m$th-order derivatives at $F$ are defined successively by 
\begin{equation}\label{def mth frechet derivative}
   \|T(G)-T(F)-\sum_{i=1}^m\frac{1}{m!}T(F;\underbrace{G-F,\cdots,G-F}_{m\text{ times}})\|=o(\|G-F\|^m) 
\end{equation}
if they exist.

\subsection{Kernel Mean Embeddings}\label{appendix: rkhs}
Suppose $\mathcal{X}$ is a fixed nonempty set and $k:\mathcal{X}\times \mathcal{X}\to\mathbb{R}$ be a real-valued positive definite kernel function associated with the reproducing kernel Hilbert space $\mathcal{H}$. Recall the reproducing kernel property says that the point-wise evaluation of $f\in\mathcal{H}$ can be expressed by its inner product with the kernel.
Suppose $\mathcal{F}$ consists of distributions over $\mathcal{X}$. Recall the definition of kernel mean embeddings (Definition 3.1 in \cite{muandet2017kernel})
\[ \mu:\mathcal{F}\to\mathcal{H},\quad F\mapsto \int k(\mathbf{x},\cdot)\mathrm{d} F( \mathbf{x}).\]
The following lemma states the condition that a distribution $F$ must satisfy to be embedded into $\mathcal{H}$.
\begin{lemma} [Lemma 3.1, \cite{muandet2017kernel}]\label{can be embed}
    If $\mathbb{E}_{X\sim F}\sqrt{k(X,X)}<\infty$, then $\mu(F)\in\mathcal{H}$ and $\mathbb{E}_F(g(X))=\langle g,\mu(F)\rangle_\mathcal{H} $.
\end{lemma}
It is easy to see that if $F$ and $G$ can be embedded into $\mathcal{H}$, so is the linear space generated by them.
Therefore, we can equip $\mathcal{D}$ with $\|\cdot\|_\mathcal{H}$ by embedding $\mathcal{D}$ into $\mathcal{H}$. This is also know as kernel maximum mean discrepancy (kernel MMD).

The next lemma states that with high probability the empirical distribution is close to the original distribution in kernel MMD.
\begin{lemma}[Theorem 3.4, \cite{muandet2017kernel}]
    Let $k:\mathcal{X}\times\mathcal{X}\to\mathbb{R}$ be a continuous positive definite kernel on a separable topological space $\mathcal{X}$ with $\sup_{x\in\mathcal{X}}k(x,x)\le C<\infty$.
    For any $\delta\in(0,1)$, with probability at least $1-\delta$,
    \[\sqrt{n}\|F_n-F\|_{\mathcal{H}}\le \sqrt{C}+\sqrt{2C\log{\frac{1}{\delta}}}\]
\end{lemma}
For further information on kernel mean embeddings, we refer the readers to \cite{muandet2017kernel}.

\subsection{Multivariate Expansion on Banach Spaces}
In this section, we delve into the foundations of multivariate calculus within the context of Banach spaces. 
We begin by delving into the fundamentals of calculus within the framework of a Banach space \cite{chang2005methods}. Let $X$ denote a Banach space, and consider an open set $U\subset X$. Let $f:U\to\mathbb{R}$ be a map.
In parallel with our treatment of statistical functionals, we define the Fr\'echet derivatives at a point $x\in U$ as elaborated earlier (for example, Definition \ref{def: 1st frechet derivative}). We denote them by $Df(x)$, $D^2f(x)$, and $D^{m}f(x)$ respectively.

Let's now explore the extension of multivariate calculus to the setting of multiple Banach spaces. Suppose we have $m$ Banach spaces $X_1,\dots,X_m$ and denote the norm of $X_i$ as $\|\cdot\|_i$.
Then the direct product of these Banach spaces, denoted as $X:=X_1\times\cdots\times X_m$, forms a Banach space itself. This space is equipped with the direct product norm $\|(x_1,\cdots,x_m)\|:=\max_i \|x_i\|_i$, where $x_i\in X_i$ ($i\in \{1,\cdots,m\}$). 

Let $U\subset X$ be an open set. For a map $f:U\to\mathbb{R}$, we can define the Fréchet derivative at a point $x\in U$, referring to it as the \textbf{total derivative} of $f$ at $x$. We denote it by $Df(x)$.
Also, for any $x=(x_1,\cdots,x_m)\in U$ and any $i\in \{1,\cdots,m\}$,
we can define the $i$-th \textbf{partial derivative} of $f$ at $x$ to be the Fr\'echet derivative of $f$ with respect to $x_i$ while holding the other variables fixed, provided the limit exists. We denote it by $\frac{\partial f}{\partial x_j} (x)$. By definition this is a linear functional on $X_i$.

If all partial derivative exist and are continuous, then $f$ is said to be of class $C^1$. 
In this case, we can differentiate first-order partial derivatives to obtain second-order partial derivatives 
\[\frac{\partial^2 f}{\partial x_k \partial x_j} (x):=\frac{\partial}{\partial x_k}(\frac{\partial f}{\partial x_j})\]
if they exists. Continuing this way leads to higher-order partial derivatives.
A function $f$ is said to be \textbf{of class} $\mathbf{C^k}$ if all the partial derivatives of $f$ of order less than or equal to $k$ exist and are continuous.

The above definitions generalize important concepts of multivariate calculus to product Banach spaces. Many desirable properties of basic multivariate calculus also generalize to this setting. For example, we have the usual chain rule for total derivatives and the equality of mixed partial derivatives. For the readers' convenience, we state the equality for mixed partial derivatives briefly as follows.
Suppose $f$ is of class $C^k$, then the partial derivatives exist up to order $k$, and the mixed partial derivatives of any order are independent of the order of differentiation. The proofs are easily generalized from the corresponding proofs in standard multivariate calculus.

For a multivariate function, the partial derivatives is easier to calculate than the total derivatives. The good news is that we can use the partial derivatives to determine the total derivative. Specifically, we have 
\begin{theorem}\label{thm: partial and total}
   Let $f$ be differentiable at $x\in U$. Then all of the partial derivatives of $f$ at $x$ exist, and 
   \[Df(x) = (\frac{\partial f}{\partial x_1} (x),\cdots,\frac{\partial f}{\partial x_m} (x))\]
\end{theorem}
\begin{proof}
    For $v=(v_1,\cdots,v_m)$ with norm samll enough such that $x+v\in U$, let $R(v)=f(x+v)-f(x)-Df(x)(v)$. The fact that $f$ is differentiable at $x$ implies that $f(v)/\|v\|$ goes to zero as $\|v\|\to 0$. The $i$-th partial derivative of $f$ at $x$ exists because it is exactly the Fr\'echet derivative when setting $v_j=0$ for all $j\ne i$, i.e., \[\frac{\partial f}{\partial x_i} (x)(v_i)=Df(x)(v_i).\]
    As $Df(x)(v)$ is linear in $v$, we obtain the desired result.
\end{proof}
We can also generalize the multivariate Taylor's theorem to Banach spaces.
In order to express it concisely, it helps to introduce some shorthand notation. For any $n$ tuple $I=(i_1,\cdots,i_n)$ of indices, and 
\[\partial^I = \frac{\partial^m}{\partial x_{i_1}\cdots \partial x_{i_n}}\]
\[(x-a)^I = (x_{i_1}-a_{i_1})\cdots (x_{i_n}-a_{i_n}). \]
Then we have
\begin{theorem}\label{thm: multi Taylor}
    Suppose $f$ is of class $C^{k+1}$ for some $k\ge 0$, then 
    \[f(x)=f(x_0)+\sum_{i=1}^k \frac{1}{i!}\sum_{I:|I|=i}\partial^If(x_0)(x-x_0)^I+\epsilon(x),\]
    where $\epsilon(x)=\frac{1}{k!}\sum_{I:|I|=k+1}(x-x_0)^I\int_0^1(1-t)^k\partial^If(x_0+t(x-x_0))\mathrm{d}t$
\end{theorem}
\begin{proof}
    The proof can be carried out by first generalizing Theorem \ref{thm: partial and total} to higher-order derivatives and then use Theorem 1.1.10 in \cite{chang2005methods}.
\end{proof}

As an example, let $\mathcal{H}_i$ be Hilbert spaces. Consider the following function defined via the tensorized inner product
\[f(x_1,\cdots,x_m)=\langle h, x_1\otimes\cdots\otimes x_m\rangle_{\mathcal{H}_1\otimes\cdots\otimes\mathcal{H}_m}\]
then
\[\frac{\partial f}{\partial x_j} (x)(v_j)=\langle h, x_1\otimes\cdots \otimes x_{j-1}
\otimes v_j\otimes x_{j+1}\otimes\cdots\otimes x_m\rangle_{\mathcal{H}_1\otimes\cdots\otimes\mathcal{H}_m}\]
and 
\[ \frac{\partial^2 f}{\partial x_i\partial x_j}[v_1,v_2]=\left\{\begin{matrix}
 \langle h, x_1\otimes\cdots \otimes v_i\otimes\cdots \otimes v_j\otimes\cdots\otimes x_m\rangle_{\mathcal{H}_1\otimes\cdots\otimes\mathcal{H}_m} & i\ne j\\
 0 & i=j
\end{matrix}\right.\]
We can derive higher order derivatives similarly. Moreover, any derivative of order strictly larger than $m$ is equal to zero, for example,
\[\partial^{m+1} f = 0\]

Hence we have by Theorem \ref{thm: multi Taylor} that
\begin{equation}\label{equation: finite-horizon taylor}
    f(x)=f(x_0)+\sum_{i=1}^m \frac{1}{i!}\sum_{I:|I|=i}\partial^If(x_0)(x-x_0)^I
\end{equation}

\section{Proof of Main Results}\label{appendix: theory}
Following \cite{lam2022subsampling}, we first verify the improvement of our Orthogonal Bootstrap under the following smoothness assumption of our target performance measure. Without explicit statement, we assume $n_1=n_2=\cdots=n_m=n$ for simplicity.

\begin{assumption}[Smoothness at True Input Model, Assumption 3 \cite{lam2022subsampling}]\label{assumption:true smoothness}
$$
    \phi( F_1,\cdots, F_m)=\phi(\hat F_1,\cdots,\hat F_m)+\sum_{i=1}^m\int \phi_i(x)\mathrm{d} \hat F_i(x)+\delta
$$
satisfies $\mathbb{E}[\delta^2]=o(n^{-1})$, $\Variance_{X_i\sim F_i} [\phi_i(X_i)]>0$, and $\mathbb{E}_{X_i\sim F_i}[\phi_i^4(X_i)]<\infty$ for all $i\in [m]$.
\end{assumption}
This assumption guarantees the existence of non-degenerate influence functions with respect to the true model.
The first condition $\mathbb{E}[\delta^2]=o(n^{-1})$ guarantees that the error of the approximation by influence functions is negligible when the number of data $n$ is large. Indeed, the influence function term is asymptotically of order $\Theta_p(n^{-\frac{1}{2}})$ by the central limit theorem, whereas the error $\delta$ is implied to be $o_p(n^{-\frac{1}{2}})$.
The second condition $\Variance_{X_i\sim F_i} [\phi_i(X_i)]>0$ says that the influence function $\phi_i$ are non-degenerate. The last assumption is needed to control the variance of the influence function at the empirical model, see Corollary \ref{cor:nonortho order}.

\begin{assumption}[Smoothness at Empirical Input Model, Assumption 4 \cite{lam2022subsampling}]\label{assumption:smoothness}
$$
    \phi(\hat F_1^b,\cdots,\hat F_m^b)=\phi(\hat F_1,\cdots,\hat F_m)+\sum_{i=1}^m\int \mathcal{I}_i^\phi(x)\mathrm{d} \hat F_i^b(x)+\epsilon
$$
satisfies $\mathbb{E}_*[\epsilon^4]=O_p(n^{-4})$ and $\mathbb E [(\mathcal{I}_i^\phi-\phi_i)^4(X_{i,1})]=o(1)$ for all $i\in [m]$.
\end{assumption}
This assumption guarantees the existence of non-degenerate influence functions with respect to the empirical model. The moment condition on the remainder is needed for controlling the variance of our orthogonal bootstrap estimator. Since a particular empirical model is chosen from the set of empirical models generated by the true model, the condition is described in terms of stochastic order. If the performance measure is sufficiently smooth, for example second order von Mises differential exists, then the second order differential is of order $\Theta_p(\frac{1}{n})$. Therefore expecting $\epsilon$ to be of order $O_p(\frac{1}{n})$ is reasonable. The last assumption entails the observation that the empirical distributions $\hat F_i$ converges to true ones $F_i$ as the data size $n$ grows and hence the empirical influence functions $\mathcal{I}_i^\phi$ are expect to approach the influence functions $\phi_i$ associated with the true input distributions.

\subsection{Variance of Non-Orthogonal Part}
When the influence functions at the empirical model exist, we can obtain a closed form formula for the variance of the non-orthogonal part. As this result is used in the algorithm, we provide the full result and do not assume $n_1=\cdots=n_m=n$ here.
\begin{lemma}\label{lem: 1st var}
If the influence functions exist and are $1$-linear, i.e.
$$\nabla \phi(\hat F_1,\cdots,\hat F_m)\left[\hat F_{1}^{b}-\hat F_1,\cdots,\hat F_{m}^{b}-\hat F_m\right]=\sum_{i=1}^m\int \mathcal{I}_i^\phi(x)\mathrm{d} \hat F_i^b(x),$$
then we have
{
\begin{equation}
        \Variance_* \left(\nabla \phi(\hat F_1,\cdots,\hat F_m)\left[\hat F_{1}^{b}-\hat F_1,\cdots,\hat F_{m}^{b}-\hat F_m\right]\right)=\sum_{i=1}^m\frac{1}{n_i^2}\sum_{j=1}^{n_i}(\mathcal{I}_i^\phi (X_{i,j}))^2
\end{equation}
}
\end{lemma}
\begin{proof}
    Notice that $\mathbb{E}_{X_i\sim \hat F_i^b}\mathcal{I}_i^\phi(X_i)=\frac{1}{n_i}\sum_{k=1}^{n_i} \mathcal{I}_i^\phi(X^b_{i,k})$, where $X^b_{i,k}\sim \hat{F_i}$. Thus we have
    \begin{align*}
        & \Variance_* (\sum_{i=1}^{m}\mathbb{E}_{X_i\sim \hat F_i^b}\mathcal{I}_i^\phi(X_i)) = \sum_{i_1=1}^{m} \sum_{i_2=1}^{m} \text{Cov}_* (\mathbb{E}_{X_{i_1}\sim \hat F_{i_1}^b}\mathcal{I}_{i_1}^\phi(X_{i_1}),\mathbb{E}_{X_{i_2}\sim \hat F_{i_2}^b}\mathcal{I}_{i_2}^\phi(X_{i_2}))\\
        &= \sum_{i=1}^{m} \text{Var}_* (\mathbb{E}_{X_{i}\sim \hat F_{i}^b}\mathcal{I}_{i_1}^\phi(X_i)) + \sum_{i\ne j}\text{Cov}_* (\mathbb{E}_{X_{i}\sim \hat F_{i}^b}\mathcal{I}_{i}^\phi(X_i),\mathbb{E}_{X_{j}\sim \hat F_{j}^b}\mathcal{I}_{j}^\phi(X_{j}))
    \end{align*}
    As $\hat F_{i}^b$ is independent of $\hat F_{j}^b$, the second term in the above equation equals to zero. Therefore,
    \begin{align*}
        &\Variance_* \left(\nabla \phi(\hat F_1,\cdots,\hat F_m)\left[\hat F_{1}^{b}-\hat F_1,\cdots,\hat F_{m}^{b}-\hat F_m\right]\right) =\sum_{i=1}^{m} \Variance_* (\mathbb{E}_{X_{i}\sim \hat F_{i}^b}\mathcal{I}_{i}^\phi(X_i))\\&=\sum_{i=1}^{m}\frac{1}{n_i^2}\sum_{k,j} \text{Cov}_*( \mathcal{I}_i^\phi(X^b_{i,k}), \mathcal{I}_i^\phi(X^b_{i,j}))=\sum_{i=1}^{m}\frac{1}{n_i}\Variance_*(\mathcal{I}_i^\phi (X_{i,1}^b))
    \end{align*}
    where the last equality holds by the independence between $X^b_{i,k}$ and $X^b_{i,j}$ $(j\ne k)$.
    Note that as $\mathbb E_*(\mathcal{I}_i^\phi (X_{i,1}^b))=0$ by definition, \[\text{Var}_*(\mathcal{I}_i^\phi (X_{i,1}^b))=\frac{1}{n_i}\sum_{j=1}^n(\mathcal{I}_i^\phi (X_{i,j}))^2,\]
    thus we obtain the desired result.
\end{proof}

From now on we continue to assume $n_1=\cdots=n_m=n$. We aim to determine the order of the random variable $\Variance_* \left(\nabla \phi(\hat F_1,\cdots,\hat F_m)\left[\hat F_{1}^{b}-\hat F_1,\cdots,\hat F_{m}^{b}-\hat F_m\right]\right)$ under Assumption \ref{assumption:true smoothness} and Assumption \ref{assumption:smoothness}. As the random variable is always positive, we use the first moment method, \emph{i.e.} the Markov's inequality.


Now we can prove the following two results.
\begin{corollary}\label{cor:nonortho order}
Under Assumption \ref{assumption:true smoothness} and Assumption \ref{assumption:smoothness}, we have
{
\begin{equation}
        \Variance_* \left(\nabla \phi(\hat F_1,\cdots,\hat F_m)\left[\hat F_{1}^{b}-\hat F_1,\cdots,\hat F_{m}^{b}-\hat F_m\right]\right)=\Theta_p(\frac{1}{n})
\end{equation}
}
\end{corollary}
\begin{proof}
    Let $$\xi=\Variance_* \left(\nabla \phi(\hat F_1,\cdots,\hat F_m)\left[\hat F_{1}^{b}-\hat F_1,\cdots,\hat F_{m}^{b}-\hat F_m\right]\right),$$ then $\xi>0$ is a random variable with respect to the probability distributions $F_1,\cdots,F_m$. The expectation of $\xi$ is $\Theta(\frac{1}{n})$ as
    \[\mathbb E \xi=\frac{1}{n}\sum_{i=1}^m\mathbb E(\mathcal{I}_i^\phi (X_{i,1}))^2,\]
    where we combine $0<\Variance_{X_i\sim F_i} [\phi_i(X_i)]<\infty$ and $\mathbb E [(\mathcal{I}_i^\phi-\phi_i)^2(X_{i,1})]=o(1)$ to show that $0<\mathbb E(\mathcal{I}_i^\phi (X_{i,1}))^2<\infty$.
    
    Via a standard first moment argument using Markov's inequality, we obtain $\xi=\Theta_p(\frac{1}{n})$. 
\end{proof}

\begin{lemma}\label{lem: 1st fourth}
Under Assumption \ref{assumption:true smoothness} and Assumption \ref{assumption:smoothness}, we have
{
\begin{equation}
         \Variance_* \left(\nabla \phi(\hat F_1,\cdots,\hat F_m)\left[\hat F_{1}^{b}-\hat F_1,\cdots,\hat F_{m}^{b}-\hat F_m\right]\right)^2=\Theta_p(\frac{1}{n^2})
\end{equation}
}
\end{lemma}
\begin{proof}
    Let $$\eta=\mathbb E_* \left(\nabla \phi(\hat F_1,\cdots,\hat F_m)\left[\hat F_{1}^{b}-\hat F_1,\cdots,\hat F_{m}^{b}-\hat F_m\right]\right)^4,$$ then $\eta>0$ is a random variable with respect to the probability distributions $F_1,\cdots,F_m$.
    Recall that \[\nabla \phi(\hat F_1,\cdots,\hat F_m)\left[\hat F_{1}^{b}-\hat F_1,\cdots,\hat F_{m}^{b}-\hat F_m\right]=\frac{1}{n}\sum_{i=1}^m\sum_{k=1}^n\mathcal{I}_{i}^\phi (X_{i,k}^b).\]
    Note that $\mathbb E_{X_i\sim \hat F_i}\mathcal{I}_i^\phi(X_i)=0$, we have 
    \begin{align*}
        \eta =&\frac{1}{n^4}\mathbb E_*\sum_{i_1,k_1,i_2,k_2,i_3,k_3,i_4,k_4}\mathcal{I}_{i_1}^\phi (X_{i_1,k_1}^b)\mathcal{I}_{i_2}^\phi (X_{i_2,k_2}^b)\mathcal{I}_{i_3}^\phi (X_{i_3,k_3}^b)\mathcal{I}_{i_4}^\phi (X_{i_4,k_4}^b)\\
        =&\frac{1}{n^4}\mathbb E_*\sum_{i,k_1,k_2,k_3,k_4}\mathcal{I}_{i}^\phi (X_{i,k_1}^b)\mathcal{I}_{i}^\phi (X_{i,k_2}^b)\mathcal{I}_{i}^\phi (X_{i,k_3}^b)\mathcal{I}_{i}^\phi (X_{i,k_4}^b)\\&+\frac{3}{n^4}\sum_{i_1\ne i_2,k_1,k_2,k_3,k_4}\mathbb E_*\mathcal{I}_{i_1}^\phi (X_{i_1,k_1}^b)\mathcal{I}_{i_1}^\phi (X_{i_1,k_2}^b)\mathbb E_*\mathcal{I}_{i_2}^\phi (X_{i_2,k_3}^b)\mathcal{I}_{i_2}^\phi (X_{i_2,k_4}^b)\\
        =&\frac{3(n-1)}{n^3}\sum_{i=1}^{m}(\mathbb E_*(\mathcal{I}_i^\phi (X_{i,1}^b))^2)^2 + \frac{1}{n^3}\sum_{i=1}^{m}\mathbb E_*(\mathcal{I}_i^\phi (X_{i,1}^b))^4 +\frac{3}{n^2}\sum_{i\ne j}\mathbb E_*(\mathcal{I}_i^\phi (X_{i,1}^b))^2\mathbb E_*(\mathcal{I}_j^\phi (X_{j,1}^b))^2
    \end{align*}
    by a calculation similar to Lemma \ref{lemma:variance of sample covariance}. Thus
    \begin{align*}
        &\Variance_* \left(\nabla \phi(\hat F_1,\cdots,\hat F_m)\left[\hat F_{1}^{b}-\hat F_1,\cdots,\hat F_{m}^{b}-\hat F_m\right]\right)^2\\=&\eta-\xi^2\\
        =&\frac{(2n-3)}{n^3}\sum_{i=1}^{m}(\mathbb E_*(\mathcal{I}_i^\phi (X_{i,1}^b))^2)^2 + \frac{1}{n^3}\sum_{i=1}^{m}\mathbb E_*(\mathcal{I}_i^\phi (X_{i,1}^b))^4 +\frac{2}{n^2}\sum_{i\ne j}\mathbb E_*(\mathcal{I}_i^\phi (X_{i,1}^b))^2\mathbb E_*(\mathcal{I}_j^\phi (X_{j,1}^b))^2.
    \end{align*}

    As $\mathbb{E}_{X_i\sim F_i}[\phi_i^4(X_i)]<\infty$ and $\mathbb E_{X_i\sim F_i} [(\mathcal{I}_i^\phi-\phi_i)^4(X_i)]=o(1)$, we have
    \begin{align*}
        &\mathbb E [\mathbb E_*(\mathcal{I}_i^\phi (X_{i,1}^b))^2]^2=\mathbb E [\frac{1}{n}\sum_{j=1}^n(\mathcal{I}_i^\phi (X_{i,j}))^2]^2\\
        =&\frac{1}{n^2}\mathbb E \sum_{j,k}(\mathcal{I}_i^\phi (X_{i,j}))^2(\mathcal{I}_i^\phi (X_{i,k}))^2=O(1),
    \end{align*}
    $\mathbb E(\mathcal{I}_i^\phi (X_{i}))^2=O(1)$, and $\mathbb E(\mathcal{I}_i^\phi (X_{i}))^4=O(1)$.

     The expectation of $\Variance_* \left(\nabla \phi(\hat F_1,\cdots,\hat F_m)\left[\hat F_{1}^{b}-\hat F_1,\cdots,\hat F_{m}^{b}-\hat F_m\right]\right)^2$ is $\Theta(\frac{1}{n^2})$ as
    \begin{align*}
        &\mathbb E \Variance_* \left(\nabla \phi(\hat F_1,\cdots,\hat F_m)\left[\hat F_{1}^{b}-\hat F_1,\cdots,\hat F_{m}^{b}-\hat F_m\right]\right)^2\\
        =& \frac{(2n-3)}{n^3}\sum_{i=1}^{m}\mathbb E(\mathbb E_*(\mathcal{I}_i^\phi (X_{i,1}^b))^2)^2 +\frac{2}{n^2}\mathbb E\sum_{i\ne j}\mathbb E_*(\mathcal{I}_i^\phi (X_{i,1}^b))^2\mathbb E_*(\mathcal{I}_j^\phi (X_{j,1}^b))^2+ \frac{1}{n^3}\sum_{i=1}^{m}\mathbb E\mathbb E_*(\mathcal{I}_i^\phi (X_{i,1}^b))^4 \\
        =&\frac{(2n-3)}{n^3}\sum_{i=1}^{m}\mathbb E (\mathbb E_*(\mathcal{I}_i^\phi (X_{i,1}^b))^2)^2+\frac{2}{n^2}\sum_{i\ne j}\mathbb E(\mathcal{I}_i^\phi (X_{i}))^2\mathbb E(\mathcal{I}_j^\phi (X_{j}))^2+\frac{1}{n^3}\sum_{i=1}^{m}\mathbb E(\mathcal{I}_i^\phi (X_{i}))^4.
    \end{align*}

    Combining the above argument and using Markov's inequality again, we get $\eta=\Theta_p(\frac{1}{n^2})$. 
\end{proof}

\subsection{Improvement of Orthogonal Bootstrap when Simulating the Mean}
\begin{theorem}\label{thm: debias}
    Under Assumption \ref{assumption:smoothness}, consider the Orthogonal Bootstrap debiasing estimator defined in Equation \eqref{eq:debiasestimator}
    \[X:=2\hat{\phi}-\frac{1}{B}\sum_{b=1}^{B}(\hat{\phi}^b-\sum_{i=1}^{m}\frac{1}{n}\sum_{j=1}^{n}\mathcal{I}_i^\phi(x_{i,j}^b)).\]
    If the number of Monte Carlo replications $B\geq Cn^\alpha$ for some absolute constant $C>0$ and $\alpha\geq 0$, then the simulation error $\Variance_*(X)=O_p(\frac{1}{n^{2+\alpha}})$.
\end{theorem}
\begin{proof}
    The simulation variance of our debiasing estimator is
    \[\Variance_* X = \frac{1}{B}\Variance_* (\hat{\phi}^b-\sum_{i=1}^{m}\frac{1}{n}\sum_{j=1}^{n}\mathcal{I}_i^\phi(x_{i,j}^b)).\]
    By Assumption \ref{assumption:smoothness}, $\hat{\phi}^b=\hat{\phi}+\sum_{i=1}^{m}\frac{1}{n}\sum_{j=1}^{n}\mathcal{I}^\phi_i(x^b_{i,j})+\epsilon$ where $\Variance_*\epsilon=O_p(\frac{1}{n^2})$.
    Therefore \[\Variance_* X=\frac{1}{B}\Variance_*\epsilon = O_p(\frac{1}{n^2B}) = O_p(\frac{1}{n^{2+\alpha}}).\]
\end{proof}
It is also easy to obtain a conditional central limit theorem for our estimator if we further assume $\Variance_*\epsilon>0$.
\begin{theorem}
    Under Assumption \ref{assumption:smoothness}, and further assuming that $\Variance_*\epsilon>0$. Consider the Orthogonal Bootstrap debiasing estimator defined in Equation \eqref{eq:debiasestimator}
    \[X:=2\hat{\phi}-\frac{1}{B}\sum_{b=1}^{B}(\hat{\phi}^b-\sum_{i=1}^{m}\frac{1}{n}\sum_{j=1}^{n}\mathcal{I}_i^\phi(x_{i,j}^b)).\]
    Then conditioning on the input data $\hat F_1,\cdots,\hat F_m$, we have
    $$\sqrt{B}\frac{X-\mathbb E_*X}{\sqrt{\Variance_* \epsilon}}\to \mathcal{N}(0,1)$$
    as $B$ tends to infinity. Moreover, we have $\Variance_*\epsilon=O_p(\frac{1}{n^2})$.
\end{theorem}

For comparison, we provide the result for bootstrap here.
\begin{theorem}\label{thm: debias sb}
    Under Assumption \ref{assumption:true smoothness} and Assumption \ref{assumption:smoothness}, consider the Standard Bootstrap debiasing estimator defined by
    \[X:=2\hat{\phi}-\frac{1}{B}\sum_{b=1}^{B}\hat{\phi}^b.\]
    Then conditioning on the input data $\hat F_1,\cdots,\hat F_m$, we have 
    $$\sqrt{B}\frac{X-\mathbb E_*X}{\sqrt{\Variance_* \hat{\phi}^b}}\to \mathcal{N}(0,1)$$
    as $B$ tends to infinity. Moreover, we have $\Variance_*\hat{\phi}^b=\Theta_p(\frac{1}{n})$.
\end{theorem}
\begin{proof}
    The simulation variance of $X$ is $\Variance_* X = \frac{1}{B}\Variance_* \hat{\phi}^b.$ By Assumption \ref{assumption:smoothness}, we have
    $\hat{\phi}^b=\hat{\phi}+\sum_{i=1}^{m}\frac{1}{n}\sum_{j=1}^{n}\mathcal{I}^\phi_i(x^b_{i,j})+\epsilon$ where $\Variance_*\epsilon=O_p(\frac{1}{n^2})$.
    Combined with Corollary \ref{cor:nonortho order}, we have
    \[\Variance_*\left[\sum_{i=1}^{m}\frac{1}{n}\sum_{j=1}^{n}\mathcal{I}^\phi_i(x^b_{i,j})\right]=\Theta_p(\frac{1}{n}).\]
    Therefore by Cauchy's inequality 
    \[\Variance_*(\sum_{i=1}^{m}\frac{1}{n}\sum_{j=1}^{n}\mathcal{I}^\phi_i(x^b_{i,j})+\epsilon) = \Theta_p(\frac{1}{n})+o_p(\frac{1}{n})=\Theta_p(\frac{1}{n}).\]
    An application of central limit theorem yields the result.  
\end{proof}
These two theorems show that conditioning on the input data, the variance of the Standard Bootstrap estimator  due to simulation is at least $n$ times larger than the variance of the Orthogonal Bootstrap estimator due to simulation. This provides a separation result for the two estimators, showing that our estimator is strictly better under certain smoothness assumptions of the performance measure.

\subsection{Improvement of Orthogonal Bootstrap when Simulating the Variance}
We first begin with a lemma that calculate the variance of the sample covariance matrix. 
\begin{lemma}\label{lemma:variance of sample covariance}
    Let $X_1,\cdots,X_n$ be i.i.d. sample drawn from a multivariate distribution function $F$ in $\mathbb{R}^p$ with finite fourth moments. 
    Let $X_{ki}$ denote the $k$-th coordinate of the random variable $X_i$, and let $\mu_k:=\mathbb{E}X_{ki}$ be its mean. Let $\mu_{kl}:=\mathbb{E}(X_{ki}-\mu_k)(X_{li}-\mu_l)$ be the second order central moment. Let $\mu_{klm}:=\mathbb{E}(X_{ki}-\mu_k)(X_{li}-\mu_l)(X_{mi}-\mu_m)$ be the third order central moment.
    Let $\mu_{klmq}:=\mathbb{E}(X_{ki}-\mu_k)(X_{li}-\mu_l)(X_{mi}-\mu_m)(X_{qi}-\mu_q)$ be the fourth order central moment.
    
    Let the sample covariance matrix be 
    \[ S_n=\frac{1}{n}\sum_{i=1}^{n}(X_i-\bar{X})(X_i-\bar{X})^T,\] 
    then \[\text{Cov}(S_{nkl},S_{nmq})=\frac{(n-1)^2}{n^3}\alpha_{klmq}+\frac{n-1}{n^3}(\mu_{km}\mu_{lq}+\mu_{kq}\mu_{ml}),\] where 
    \[\alpha_{klmq}=\mu_{klmq}-\mu_{kl}\mu_{mq}.\]
\end{lemma}
\begin{proof}
    First of all we have the decomposition
    $$S_{nkl}=\underbrace{\frac{1}{n}\sum_{i=1}^{n}(X_{ki}-\mu_k)(X_{li}-\mu_l)}_{:=A_{nkl}}-\underbrace{(\bar X_k-\mu_k)(\bar X_l-\mu_l)}_{:=B_{nkl}}.$$
    Note that the covariance between $S_{nkl}$ and $S_{nmq}$ is 
    \begin{align*}
        &\text{Cov}(S_{nkl},S_{nmq})\\
        =&\text{Cov}(A_{nkl},A_{nmq})-\text{Cov}(A_{nkl},B_{nmq})-\text{Cov}(B_{nkl},A_{nmq})+\text{Cov}(B_{nkl},B_{nmq})
    \end{align*}
    For the first term we have 
    \begin{align*}
        &\text{Cov}(A_{nkl},A_{nmq})\\
        =&\frac{1}{n}\text{Cov}((X_{ki}-\mu_k)(X_{li}-\mu_l),(X_{mi}-\mu_m)(X_{qi}-\mu_q))\\
        =&\frac{1}{n}\left(\mathbb{E}(X_{ki}-\mu_k)(X_{li}-\mu_l)(X_{mi}-\mu_m)(X_{qi}-\mu_q)-\mu_{kl}\mu_{mq}\right)\\
        =&\frac{1}{n}\alpha_{klmq}
    \end{align*}
    For the last term we have
    \begin{align*}
        &\text{Cov}(B_{nkl},B_{nmq})\\
        =&\text{Cov}(\frac{1}{n^2}\sum_i(X_{ki}-\mu_k)\sum_j(X_{lj}-\mu_l),\frac{1}{n^2}\sum_s(X_{ms}-\mu_k)\sum_t(X_{qt}-\mu_q))\\
        =& \frac{1}{n^4}\sum_{i,j,s,t}\mathbb E(X_{ki}-\mu_k)(X_{lj}-\mu_l)(X_{ms}-\mu_m)(X_{qt}-\mu_q)\\&-\frac{1}{n^4}\mathbb E \sum_{i,j}(X_{ki}-\mu_m)(X_{lj}-\mu_l)\mathbb E\sum_{s,t}(X_{ms}-\mu_m)(X_{qt}-\mu_q)\\
        =&\frac{n(n-1)}{n^4}(\mu_{kl}\mu_{mq}+\mu_{km}\mu_{lq}+\mu_{kq}\mu_{ml}) -\frac{1}{n^2}\mu_{kl}\mu_{mq} \\
        &+\frac{1}{n^3}\mathbb E(X_{k}-\mu_k)(X_{l}-\mu_l)(X_{m}-\mu_k)(X_{q}-\mu_q)\\
        =&\frac{n-1}{n^3}(\mu_{km}\mu_{lq}+\mu_{kq}\mu_{ml})+\frac{1}{n^3}\alpha_{klmq}
    \end{align*}
    For the cross term we have 
    \begin{align*}
        &\text{Cov}(A_{nkl},B_{nmq})\\
        =&\text{Cov}(\frac{1}{n}\sum_{i}(X_{ki}-\mu_k)(X_{li}-\mu_l),\frac{1}{n^2}\sum_s(X_{ms}-\mu_k)\sum_t(X_{qt}-\mu_q))\\
        =& \frac{1}{n^3}\sum_{i,s,t}\mathbb E(X_{ki}-\mu_k)(X_{li}-\mu_l)(X_{ms}-\mu_m)(X_{qt}-\mu_q)\\&-\frac{1}{n^3}\mathbb E \sum_{i}(X_{ki}-\mu_k)(X_{li}-\mu_l)\mathbb E\sum_{s,t}(X_{ms}-\mu_s)(X_{qt}-\mu_q)\\
        =&\frac{n(n-1)}{n^3}\mu_{kl}\mu_{mq}+\frac{1}{n^2}\mu_{klmq} -\frac{1}{n}\mu_{kl}\mu_{mq} \\
        =&\frac{1}{n^2}\alpha_{klmq}
    \end{align*}
    Therefore 
    \[\text{Cov}(S_{nkl},S_{nmq})=\frac{(n-1)^2}{n^3}\alpha_{klmq}+\frac{n-1}{n^3}(\mu_{km}\mu_{lq}+\mu_{kq}\mu_{ml})\]
\end{proof}

\begin{theorem}\label{thm: original ci}
    Under Assumption \ref{assumption:smoothness}, consider the Orthogonal Bootstrap variance estimator defined in Equation \eqref{eq:varianceestimator}
    \[ X:=\frac{1}{n^2}\sum_{i=1}^m\sum_{j=1}^n(\mathcal{I}_i^\phi (X_{i,j}))^2+\frac{1}{B}\sum_{b=1}^{B}\left(\hat\phi^b-\mathcal{\hat I}^b-\overline{\phi-\mathcal{I}}\right)^2+\frac{2}{B}\sum_{b=1}^{B}\left(\hat\phi^b-\mathcal{\hat I}^b-\overline{\phi-\mathcal{I}}\right)\left(\mathcal{\hat I}^b-\overline{\mathcal{I}}\right).\] 
    If the number of Monte Carlo replications $B\geq Cn^\alpha$ for some absolute constant $C>0$ and $\alpha\geq 0$, then $\Variance_*(X)=O_p(\frac{1}{n^{3+\alpha}})$.
\end{theorem}
\begin{proof}
Let $Y^b:=\hat\phi^b-\mathcal{\hat I}^b$ and $Z^b:=\hat\phi^b$. Regard $(Y^b,Z^b)$ as a two dimensional vector, then its sample covariance matrix is 
    
    \[S_{B}=\begin{pmatrix}
\frac{1}{B}\sum_{b=1}^{B}\left(\hat\phi^b-\mathcal{\hat I}^b-\overline{\phi-\mathcal{I}}\right)^2  &\frac{1}{B}\sum_{b=1}^{B}\left(\hat\phi^b-\mathcal{\hat I}^b-\overline{\phi-\mathcal{I}}\right)\left(\mathcal{\hat I}^b-\overline{\mathcal{I}}\right) \\
 \frac{1}{B} \sum_{b=1}^{B}\left(\hat\phi^b-\mathcal{\hat I}^b-\overline{\phi-\mathcal{I}}\right)\left(\mathcal{\hat I}^b-\overline{\mathcal{I}}\right) &\frac{1}{B}\sum_{b=1}^{B}\left(\mathcal{\hat I}^b-\overline{\mathcal{I}}\right)^2
\end{pmatrix},\]
    so $S_{B11}:=\frac{1}{B}\sum_{b=1}^{B}\left(\hat\phi^b-\mathcal{\hat I}^b-\overline{\phi-\mathcal{I}}\right)^2$ and $S_{B12}:=\frac{2}{B}\sum_{b=1}^{B}\left(\hat\phi^b-\mathcal{\hat I}^b-\overline{\phi-\mathcal{I}}\right)\left(\mathcal{\hat I}^b-\overline{\mathcal{I}}\right)$.
    Then the variance of $X$ is
    \[\Variance_* X =\Variance_*(S_{B11}+S_{B12})=\Variance_* S_{B11}+2\text{Cov}_*(S_{B11},S_{B12})+\Variance_* S_{B12}.\]

    By Lemma \ref{lemma:variance of sample covariance}
    \[\Variance_* S_{B11}=\frac{(B-1)^2}{B^3}\alpha_{1111}+\frac{2(B-1)}{B^3}\mu_{11}^2,\]
    \[\Variance_* S_{B12}=\frac{(B-1)^2}{B^3}\alpha_{1212}+\frac{2(B-1)}{B^3}\mu_{12}^2,\]
    and
    \[\text{Cov}(S_{B11},S_{B12})=\frac{(B-1)^2}{B^3}\alpha_{1112}+\frac{2(B-1)}{B^3}\mu_{11}\mu_{12}.\]
    Therefore 
    \[\Variance_* X=\frac{(B-1)^2}{B^3}(\alpha_{1111}+\alpha_{1212}-2\alpha_{1112})+\frac{2(B-1)}{B^3}(\mu_{11}-\mu_{12})^2
    \]

    By definition and recalling that $\mathbb E_*\mathcal{\hat I}^b=0$, $\Variance_*\mathcal{\hat I}^b=\Theta_p(\frac{1}{n})$ (Corollary \ref{cor:nonortho order}), and $\Variance_*(\mathcal{\hat I}^b)^2=\Theta_p(\frac{1}{n^2})$ (Lemma \ref{lem: 1st fourth}),
    \begin{align*}
        &\mu_{11}=\Variance_* \epsilon^b = O_p(\frac{1}{n^2})\\
        &\mu_{12}=\Variance_* \epsilon^b +\mathbb E_*\epsilon^b\mathcal{\hat I}^b= O_p(\frac{1}{n^{\frac{3}{2}}})\\
        &\mu_{1111}=\mathbb E_*(\epsilon-\mathbb E_*\epsilon)^4=O_p(\frac{1}{n^4})\\
        &\mu_{1212}=\mathbb E_*(\epsilon-\mathbb E_*\epsilon)^4+2\mathbb E_*(\epsilon-\mathbb E_*\epsilon)^3\mathcal{\hat I}^b+\mathbb E_*(\epsilon-\mathbb E_*\epsilon)^2(\mathcal{\hat I}^b)^2=O_p(\frac{1}{n^3})\\
        &\mu_{1112}=\mathbb E_*(\epsilon-\mathbb E_*\epsilon)^4+\mathbb E_*(\epsilon-\mathbb E_*\epsilon)^3\mathcal{\hat I}^b=O_p(\frac{1}{n^{\frac{7}{2}}})
    \end{align*}
    Therefore when $B\geq Cn^\alpha$, $\Variance_* X=O_p(\frac{1}{n^{3+\alpha}})$.
\end{proof}
\begin{theorem}
    Under Assumption \ref{assumption:smoothness}, consider the Orthogonal Bootstrap variance estimator defined in Equation \eqref{eq:varianceestimator}
    \[ X:=\frac{1}{n^2}\sum_{i=1}^m\sum_{j=1}^n(\mathcal{I}_i^\phi (X_{i,j}))^2+\frac{1}{B}\sum_{b=1}^{B}\left(\hat\phi^b-\mathcal{\hat I}^b-\overline{\phi-\mathcal{I}}\right)^2+\frac{2}{B}\sum_{b=1}^{B}\left(\hat\phi^b-\mathcal{\hat I}^b-\overline{\phi-\mathcal{I}}\right)\left(\mathcal{\hat I}^b-\overline{\mathcal{I}}\right).\] 
    Then conditioning on input data $\hat F_1,\cdots,\hat F_m$, we have
    $$\sqrt{B}\frac{X-\mathbb E_*X}{\sqrt{S^2}}\to \mathcal{N}(0,1)$$
    as $B$ tends to infinity. Moreover, $S^2=O_p(\frac{1}{n^3})$.
\end{theorem}

For comparison, we provide the result for Standard Bootstrap here.
\begin{theorem}\label{thm: original ci sb}
    Under Assumption \ref{assumption:true smoothness} and Assumption \ref{assumption:smoothness}, consider the Standard Bootstrap variance estimator
    \[X:=\frac{1}{B}\sum_{b=1}^{B}(\hat{\phi}^b-\overline{\phi})^2.\]
    Then conditioning on the input data $\hat F_1,\cdots,\hat F_m$, we have
    $$\sqrt{B}\frac{X-\mathbb E_*X}{\sqrt{S^2}}\to \mathcal{N}(0,1).$$
    Moreover, we have $S^2=\Theta_p(\frac{1}{n^2})$.
\end{theorem}
\begin{proof}

     By Lemma \ref{lemma:variance of sample covariance}, the variance of $X$ is
    \[\Variance_* X = \frac{(B-1)^2}{B^3}\left[\Variance_*(\hat\phi^b-E_*\hat\phi^b)^2\right]+\frac{2(B-1)}{B^3}\Variance^2_*\hat\phi^b.\]
    
    Now we need to specify $\Variance_*(\hat\phi^b-E_*\hat\phi^b)^2$.
    By Assumption \ref{assumption:smoothness},
    \[\hat{\phi}^b=\hat{\phi}+\mathcal{\hat I}^b+\epsilon^b\]
    where $\mathcal{\hat I}^b=\sum_{i=1}^{m}\frac{1}{n}\sum_{j=1}^{n}\mathcal{I}^\phi_i(x^b_{i,j})$ and $\mathbb E_*(\epsilon^b)^4=O_p(\frac{1}{n^4})$. 
    We have already obtained
    \[\Variance_*(\mathcal{\hat I}^b) = \Theta_p(\frac{1}{n}),\]
    and \[\Variance_*(\mathcal{\hat I}^b)^2 = \Theta_p(\frac{1}{n^2}),\]
    in Corollary \ref{cor:nonortho order} and Lemma \ref{lem: 1st fourth} respectively.
    
    Combining the above estimate, we obtain $\Variance_*(\hat\phi^b-E_*\hat\phi^b)^2=\Theta_p(\frac{1}{n^2})$. Moreover, $\Variance^2_*\hat\phi^b$ is also $\Theta_p(\frac{1}{n^2})$.
    Therefore we can apply the central limit theorem when we are considering the large $B$ limit, conditioning on the input data. 
\end{proof}
Similar to the situation of debiasing, these two theorems show that conditioning on the input data, the variance of the Standard Bootstrap estimator due to simulation is at least $n$ times larger than the variance of the Orthogonal Bootstrap estimator due to simulation when estimating variance.

\subsection{Improved Variance Estimation}\label{appendix:iob}
Occasionally, our Orthogonal Bootstrap estimator for variance can yield negative values, particularly when the sample size $n$ is small. In such instances, the resulting confidence intervals constructed using the Orthogonal Bootstrap method become devoid of meaningful interpretation.

However, this issue can be effectively addressed through a slight modification to our estimator, without compromising the fundamental theoretical properties of our method. We refer to this enhanced version as the \textbf{Improved Orthogonal Bootstrap}.

The core idea behind the Improved Orthogonal Bootstrap is that when simulation results produce negative values, we seamlessly transition to utilizing the Infinitesimal Jackknife method. The algorithm for the Improved Orthogonal Bootstrap is summarized in Algorithm \ref{alg:ciconstructioniob}. Importantly, we emphasize that this improvement preserves the desirable properties of the original Orthogonal Bootstrap method.

\begin{algorithm*}
\caption{Variance Estimation via Improved Orthogonal Bootstrap}\label{alg:ciconstructioniob}
\textbf{Input}: A generic performance measure $\phi(F_1,\cdots,F_m)$, i.i.d samples  $\{X_{i,1},\cdots,X_{i,n_i}\}\in\mathbb{R}^{d_1}$ of $F_i$, and influence function $\mathcal{I}_i^\phi$ of $\phi$ respect to $\hat{F}_i$.  \\
 \textbf{Output}: Estimation of {$\Variance_{F_1,\cdots,F_m}\phi(\hat F_1,\cdots,\hat F_m)$}.
\begin{algorithmic}
\STATE $\hat \phi\leftarrow \phi(\hat F_1,\cdots,\hat F_m)$, where $\hat F_i=\frac{1}{n_i}\sum_{j=1}^{n_i}\delta_{X_{i,j}}$

\FOR{b=1:B}
\FOR{i=1:m}
\STATE Sample $\{x_{i,1}^b,\cdots,x_{i,n_i}^b\}$ i.i.d from $\hat{F_i}$
\ENDFOR
\STATE $\hat\phi^b\leftarrow\phi(\hat F_1^b,\cdots,\hat F_m^b)$, where $\hat{F}_i^b=\frac{1}{n_i}\sum_{j=1}^{n_i} \delta_{x_{i,j}^b}$
\STATE $\mathcal{\hat I}^b=\sum_{i=1}^m\frac{1}{n_i}\sum_{j=1}^{n_i}\mathcal{I}_i^\phi(\tilde{x}_{i,j})$
\ENDFOR
\STATE $\overline{\phi-\mathcal{I}}\leftarrow\frac{1}{B}\sum_{b=1}^{B}(\hat\phi^b-\mathcal{\hat I}^b)$, $\overline{\mathcal{I}}\leftarrow\frac{1}{B}\sum_{b=1}^{B}\mathcal{\hat I}^b$.
\STATE Construct the $1-\alpha$-confidence interval as $$[\hat\phi-z_{1-\alpha/2}S,\hat\phi+z_{1-\alpha/2}S],$$ where 
\begin{equation}
    S^2=\left\{\begin{matrix}
  S_1^2& \quad\text{if}\quad S_1^2\geq0\\
  S_2^2& \quad\text{if}\quad S_1^2<0
\end{matrix}\right.,
\end{equation}

\begin{equation}
    \begin{aligned}
        S_1^2&=\sum_{i=1}^m\frac{1}{n_i^2}\sum_{j=1}^{n_i}(\mathcal{I}_i^\phi (X_{i,j}))^2+\frac{1}{B}\sum_{b=1}^{B}\left(\hat\phi^b-\mathcal{\hat I}^b-\overline{\phi-\mathcal{I}}\right)^2+\frac{2}{B}\sum_{b=1}^{B}\left(\hat\phi^b-\mathcal{\hat I}^b-\overline{\phi-\mathcal{I}}\right)\left(\mathcal{\hat I}^b-\overline{\mathcal{I}}\right),
    \end{aligned}
\end{equation}
\begin{equation}
    \begin{aligned}
        S_2^2&=\sum_{i=1}^m\frac{1}{n_i^2}\sum_{j=1}^{n_i}(\mathcal{I}_i^\phi (X_{i,j}))^2,
    \end{aligned}
\end{equation}
and $z_{1-\alpha/2}$ is the $(1-\alpha/2)$-quantile of the standard normal. 
\end{algorithmic}
\end{algorithm*}

Now we assert that our improved Orthogonal Bootstrap still has the same favorable properties as original Orthogonal Bootstrap. Crucially, it's worth noting that the occurrence of negative simulation results is exceedingly rare, and as such, has a negligible impact on the overall algorithm's performance.
\begin{lemma} [Chebyshev's Inequality]\label{lem: cheby 2}
    For a random variable $X$ with finite variance $\Variance X$ and for $t>\mathbb{E}X$, we have
       \[ \mathbb{P}(X>t)\le \frac{\Variance X}{(t-\mathbb{E}X)^2}.\]
\end{lemma}
\begin{lemma}\label{lem: tail mean}
    For a random variable $X\geq0$,
    \[\mathbb{E} X = \int_0^\infty \mathbb{P}(X>t)\mathrm{d}t.\]
\end{lemma}
\begin{theorem}\label{thm: improved ci}
    Under the same assumptions as Theorem \ref{assumption:smoothness}, consider the Orthogonal Bootstrap debiasing estimator defined in Algorithm \ref{alg:ciconstructioniob}.
    If the number of  Monte Carlo replications $B\geq Cn^\alpha$ for some absolute constant $C>0$ and $\alpha\geq 0$, then $\Variance_*(X)=O_p(\frac{1}{n^{3+\alpha}})$.
\end{theorem}
\begin{proof}
    Denote $\mathcal{I}:=\frac{1}{n^2}\sum_{i=1}^m\sum_{j=1}^n(\mathcal{I}_i^\phi (X_{i,j}))^2>0$, and 
    \[X:=\frac{1}{B}\sum_{b=1}^{B}\left(\hat\phi^b-\mathcal{\hat I}^b-\overline{\phi-\mathcal{I}}\right)^2+\frac{2}{B}\sum_{b=1}^{B}\left(\hat\phi^b-\mathcal{\hat I}^b-\overline{\phi-\mathcal{I}}\right)\left(\mathcal{\hat I}^b-\overline{\mathcal{I}}\right).\]
 As $S^2>0$, we have $\mathbb{E}_*(S^2)=\mathcal{I}+\mathbb{E}_*(X)>0$ and $\mathcal{I}$ is of order $O_p(\frac{1}{n})$.
    Now our improved estimator is \[\tilde{S}:=\mathcal{I}+X 1_{X>-\mathcal{I}}.\]

For any $t\geq\mathcal{I}$, by Lemma \ref{lem: cheby 2} we have
    \[\mathbb{P}_*(X<-t)\le \frac{\Variance_* X}{(t+\mathbb{E}_*X)^2}.\]
As $-X1_{X<-\mathcal{I}}\geq 0$, by Lemma \ref{lem: tail mean},
\begin{align*}
    \mathbb{E}_*(X1_{X<-\mathcal{I}})&=-\int_{0}^{\infty} \mathbb{P}_*(X1_{X<-\mathcal{I}}<-t)\mathrm{d}t\\
    &= -\int_{0}^{\mathcal{I}} \mathbb{P}_*(X<-\mathcal{I})\mathrm{d}t -\int_{\mathcal{I}}^\infty \mathbb{P}_*(X<-t)\mathrm{d}t\\
    &\geq -\frac{\mathcal{I}\Variance_* X}{(\mathcal{I}+\mathbb{E}_*X)^2}-\frac{\Variance_* X}{\mathcal{I}+\mathbb{E}_*X}.
\end{align*}
With $\mathbb{E}_* X = O_p(\frac{1}{n^{\frac{3}{2}}})$ and $\Variance_* X=O_p(\frac{1}{Bn^3})$, we get $\mathbb{E}_* X1_{X<-\mathcal{I}}=O_p(\frac{1}{Bn^2})$. Further,
    \[\Variance_* \tilde{S} \leq \Variance_* X +(\mathbb{E}_* X)^2 - (\mathbb{E}_* X1_{X>-\mathcal{I}})^2\leq \Variance_* X +2\mathbb{E}_* X\mathbb{E}_* X1_{X<-\mathcal{I}}.\]

    As $\mathbb{E}_* X = O_p(\frac{1}{n^{\frac{3}{2}}})$ and $\mathbb{E}_* X1_{X<-\mathcal{I}}=O_p(\frac{1}{Bn^2})$, $\Variance_* \tilde{S}$ is still of order 
    $O_p(\frac{1}{Bn^3})$, which is the same as the order of $\Variance_* S$.
\end{proof}

\subsection{Verifying the Assumptions using Kernel Mean Embedding}
\label{section:verifyingKernelMMD}

Now we investigate cases where we can theoretically prove Assumption \ref{assumption:true smoothness} and Assumption \ref{assumption:smoothness}.

Let us consider the performance measure defined in Equation \eqref{equation: rkhs embed}
\begin{equation*}
     \phi(F_1,\cdots,F_m)=h(\mu_1(F_1),\cdots,\mu_m(F_m))
\end{equation*}
where $\mu_i$ are kernel mean embeddings using kernels $k_i$, $\mathcal{H}:=\mathcal{H}_1\times\cdots\times\mathcal{H}_m$ where $\mathcal{H}_i$ is the reproducing kernel Hilbert space respect to the kernel $k_i$ and $h:\mathcal{H}\to\mathbb{R}$ is a functional on $\mathcal{H}$. Denote $\mu:=\mu_1\times\cdots\mu_m$ and $F=F_1\times\cdots F_m$ for simplicity.
Suppose Assumption \ref{assumption: rkhs embed} and Assumption \ref{assumption: lip in rkhs} holds for the performance measure. For the readers convenience, we restate the two assumptions here.

\begin{assumption}
   There exists kernel mean embeddings $\mu_i:\mathcal{F}_i\to\mathcal{H}_i$ which maps $F_i$ into $\mathcal{H}_i$ for $i=1,\cdots,m$. 
    Moreover, for all $i=1,\cdots,m$, $\mathbb{E}k_i(X_i,X_i)^4<\infty$ and $\mathbb{E}k_i(X_i,Y_i)^4<\infty$ where $X_i,Y_i$ are independent samples from $F_i$.
\end{assumption}
It is easy to see if $\mathbb{E}k_i(X_i,X_i)^4<\infty$ where $X_i$ are samples from $F_i$, then $F_i$ can be embedded into RKHS with kernel $k_i$ by Lemma \ref{can be embed}.

\begin{assumption}
    The non-constant functional $h:\mathcal{H}_1\times\cdots\times\mathcal{H}_m\to\mathbb{R}$ is of class $C^1$ and its derivative is Lipschitz in the sense that 
    $|Dh(x_1)(v)-Dh(x_2)(v)|\le L\|x_1-x_2\|_{\mathcal{H}}\|v\|_{\mathcal{H}}$. Moreover, $\|\partial_ih(\mu(F))^4\|_{\mathcal{H}_i}<\infty$.
\end{assumption}

The Lipschitz condition is the key ingredient for our result. If $h$ is of class $C^1$ and every partial derivative of $h$ is Lipschitz, then by Theorem \ref{thm: partial and total} the Lipschitz condition is satisfied. The last assumption is important for controlling the moment of influence function.

\begin{theorem}\label{thm: influence exist in rkhs}
    Under Assumption \ref{assumption: rkhs embed} and Assumption \ref{assumption: lip in rkhs}, Assumption \ref{assumption:true smoothness} and Assumption \ref{assumption:smoothness} holds for the performance measure given by (\ref{equation: rkhs embed}).
\end{theorem}
\begin{proof}
First we prove Assumption \ref{assumption:true smoothness}.
Denote $x:=\mu(F)$ and $y:=\mu(\hat{F})$. Then as $h$ is of class $C^1$, we can write
    \[ h(y)=h(x)+Dh(x)[y-x]+\delta,\]
    where the remainder term can be controlled by the Lipschitz property of $Dh$. Specifically,
    \begin{align*}
        \delta &= h(y)-h(x)-Dh(x)[y-x] \\
        &= \int_0^1 Dh(x+t(y-x))\mathrm{d}t [y-x] -Dh(x)[y-x]\\
        &= \int_0^1 (Dh(x+t(y-x))-Dh(x))\mathrm{d}t [y-x] 
    \end{align*}
    so  $|\delta| \le \int_0^1 |Dh(x+t(y-x))-Dh(x)|\mathrm{d}t [y-x] \le \frac{L}{2}\|y-x\|^2_\mathcal{H} $. Now we control the term $\|y-x\|^2_\mathcal{H}$:
    \begin{align*}
        \|y-x\|_\mathcal{H}^2 &=\max_i \int\int k_i(y_1,y_2)\mathrm{d}[\hat F_i(y_1)-F_i(y_1)]\mathrm{d}[\hat F_i(y_2)-F_i(y_2)]
    \end{align*}   
    As $\mathbb{E}k_i(X_i,X_i)^4<\infty$ and $\mathbb{E}k_i(X_i,Y_i)^4<\infty$ where $X_i,Y_i$ are independent samples from $ F_i$,
    by Lemma \ref{lemma:four}, we have $\mathbb{E} \|y-x\|_\mathcal{H}^8=O(\frac{1}{n^4})$. 
    Thus $\mathbb{E} \delta^2=o(\frac{1}{n})$.

    Next, by Theorem \ref{thm: partial and total} we can write 
    $$Dh(x)(v)=\sum_{i=1}^m\partial_ih(x)(v_i).$$
    Noting that we can view a linear functional on a Hilbert space as an element in this Hilbert space via Riesz's representation theorem, we have $\partial_ih(x)(y_i)=\left\langle \partial_i h(x),y_i \right\rangle_{\mathcal{H}_i}=\mathbb E_{X_i\sim \hat F_i} \partial_i h(x)(X_i)$. So the influence function is the centered version of $\mathbb E_{X_i\sim \hat F_i} \partial_i h(x)(X_i)$, i.e. $\mathbb E_{X_i\sim \hat F_i} \partial_i h(x)(X_i)-\mathbb E_{X_i\sim  F_i} \partial_i h(x)(X_i)$. Therefore if we set $\phi_i = \partial_i h(x)-\mathbb E_{X_i\sim  F_i} \partial_i h(x)(X_i)$, we can obtain
    \[\phi(\hat F_1,\cdots,\hat F_m)=\phi( F_1,\cdots, F_m)+\sum_{i=1}^m\int \phi_i(x)\mathrm{d} \hat F_i(x)+\delta.\]

    Now the condition $\|\partial_ih(x)^4\|_{\mathcal{H}}<\infty$ implies $\mathbb E_{X_i\sim F_i}  (\phi_i(X_i))^4<\infty$, and the condition that $h$ is not constant is equivalent to $\Variance_{X_i\sim F_i} \phi_i(X_i) >0$.

Next we prove Assumption \ref{assumption:smoothness}. Denote $y:=\mu(\hat{F})$ and $z:=\mu(\hat{F}^b)$. Again we can write
    \[ h(z)=h(y)+Dh(y)[z-y]+\epsilon,\]
    where
    \begin{align*}
        \epsilon &= h(z)-h(y)-Dh(y)[z-y] \\
        &= \int_0^1 Dh(y+t(z-y))\mathrm{d}t [z-y] -Dh(y)[z-y]\\
        &= \int_0^1 (Dh(y+t(z-y))-Dh(y))\mathrm{d}t [z-y] 
    \end{align*}
    so 
    \begin{align*}
        |\epsilon| &\le \int_0^1 |Dh(y+t(z-y))-Dh(y)|\mathrm{d}t [z-y] \\
        &\le \frac{L}{2}\|z-y\|^2_\mathcal{H} 
    \end{align*}
Again we control the term $\|z-y\|^2_\mathcal{H}$:
    \begin{align*}
        \|z-y\|_\mathcal{H}^2 &=\max_i \frac{1}{n^2}\sum_{j,k}^n \left( k_i(x_{i,j}^b,x_{i,k}^b)-2k_i(x_{i,j}^b,x_{i,k})+k_i(x_{i,j},x_{i,k})\right)\\
        &=\max_i \int\int k_i(y_1,y_2)\mathrm{d}[\hat F^b_i(y_1)-\hat{F}_i(y_1)]\mathrm{d}[\hat F^b_i(y_2)-\hat{F}_i(y_2)]
    \end{align*}   
    As $\mathbb{E}k_i(X_i,X_i)^4<\infty$ and $\mathbb{E}k_i(X_i,Y_i)^4<\infty$ where $X_i,Y_i$ are independent samples from $F_i$,
    by Markov's inequality, we have $\mathbb{E}_*k_i(X_i^b,X_i^b)^4<\infty$ and $\mathbb{E}_*k_i(X_i^b,Y_i^b)^4<\infty$ with high probability where $X_i^b,Y_i^b$ are independent samples from $\hat F_i$.
    By Lemma \ref{lemma:variance} and Lemma \ref{lemma:four}, we have 
    $\mathbb{E}_* \|y-x\|_\mathcal{H}^4=O_p(\frac{1}{n^2})$ and $\mathbb{E}_* \|y-x\|_\mathcal{H}^8=O_p(\frac{1}{n^4})$. 

    Again, by Theorem \ref{thm: partial and total} we can write 
    $$Dh(y)(v)=\sum_{i=1}^m\partial_ih(y)(v_i),$$
    and we have $\partial_ih(y)(z_i)=\left\langle \partial_i h(y),z_i \right\rangle_{\mathcal{H}_i}=\mathbb E_{X_i\sim \hat F_i^b} \partial_i h(y)(X_i)$. So the influence function is the centered version of $\mathbb E_{X_i\sim \hat F_i^b} \partial_i h(y)(X_i)$, i.e. $\mathbb E_{X_i\sim \hat F_i^b} \partial_i h(y)(X_i)-\mathbb E_{X_i\sim \hat F_i} \partial_i h(y)(X_i)$. Therefore if we set $\mathcal{I}_i^\phi = \partial_i h(y)-\mathbb E_{X_i\sim \hat F_i} \partial_i h(y)(X_i)$, we can obtain
    \[\phi(\hat F_1^b,\cdots,\hat F_m^b)=\phi(\hat F_1,\cdots,\hat F_m)+\sum_{i=1}^m\int \mathcal{I}_i^\phi(x)\mathrm{d} \hat F_i^b(x)+\epsilon.\]

    Now we prove $\mathbb E [(\mathcal{I}_i^\phi-\phi_i)^4(X_{i,1})]=o(1)$. First note that
    \begin{align*}
        \quad & \mathbb E [(\mathcal{I}_i^\phi-\phi_i)^4(X_{i,1})]\\& =  \mathbb E [( \partial_i h(y)-\mathbb E_{X_i\sim \hat F_i} \partial_i h(y)(X_i)- \partial_i h(x)+\mathbb E_{X_i\sim  F_i} \partial_i h(x)(X_i))^4(X_{i,1})]\\
        &\leq 8\mathbb E (\partial_i h(y)(X_{i,1})-\partial_i h(x)(X_{i,1}))^4+8\mathbb E (\mathbb E_{X_i\sim \hat F_i} \partial_i h(y)(X_i)-\mathbb E_{X_i\sim  \hat F_i} \partial_i h(x)(X_i))^4\\\quad \quad &+8\mathbb E (\mathbb E_{X_i\sim \hat F_i} \partial_i h(x)(X_i)-\mathbb E_{X_i\sim  F_i} \partial_i h(x)(X_i))^4
    \end{align*}
    We deal with each term separately. For the first term, noting that $\mathbb{E}k_i(X_i,X_i)^4<\infty$ where $X_i \sim F_i$ and $\mathbb{E} \|y-x\|_\mathcal{H}^8=O(\frac{1}{n^4})$, we have
    \begin{align*}
        \mathbb E (\partial_i h(y)(X_{i,1})-\partial_i h(x)(X_{i,1}))^4 &\leq L^4 \mathbb E \|y-x\|^4_{\mathcal{H}} \|k_i(\cdot,X_{i,1})\|^4_{\mathcal{H}}\\
        &\leq L^4 \sqrt{\mathbb E \|y-x\|^8_{\mathcal{H}} \mathbb E\|k_i(\cdot,X_{i,1})\|^8_{\mathcal{H}}}\\
        &= L^4 \sqrt{\mathbb E \|y-x\|^8_{\mathcal{H}} \mathbb E k_i^4(X_{i,1},X_{i,1})}\\
        &= O(\frac{1}{n^2})
    \end{align*}
    For the second term, we calculate $\|y\|^2_{\mathcal{H}}=\max_i \frac{1}{n^2}\sum_{j,k}k_i(x_{i,j},x_{i,k})$. As $\mathbb{E}k_i(X_i,X_i)^4<\infty$ and $\mathbb{E}k_i(X_i,Y_i)^4<\infty$ where $X_i,Y_i$ are independent samples from $F_i$, $\mathbb E\|y\|^8_{\mathcal{H}}<\infty$. So,
    \begin{align*}
        \mathbb E (\mathbb E_{X_i\sim \hat F_i} \partial_i h(y)(X_i)-\mathbb E_{X_i\sim  \hat F_i} \partial_i h(x)(X_i))^4 &\leq L^4 \mathbb E \|y-x\|^4_{\mathcal{H}} \|y\|^4_{\mathcal{H}}\\
        &\leq L^4 \sqrt{\mathbb E \|y-x\|^8_{\mathcal{H}} \mathbb E\|y\|^8_{\mathcal{H}}}\\
        &= O(\frac{1}{n^2})
    \end{align*}

    For the last term,  noting that $\|\partial_i h(x)\|_{\mathcal{H}}^4<\infty$ and $\mathbb{E} \|y-x\|_\mathcal{H}^4=O(\frac{1}{n^2})$
    \begin{align*}
        \mathbb E (\mathbb E_{X_i\sim \hat F_i} \partial_i h(x)(X_i)-\mathbb E_{X_i\sim  F_i} \partial_i h(x)(X_i))^4 
        &=\mathbb E (\partial_i h(x)(y-x))^4 \\
        &\leq \|\partial_i h(x)\|_{\mathcal{H}}^4\mathbb E \|y-x\|_{\mathcal{H}}^4\\
        &= O(\frac{1}{n^2})
    \end{align*}
    
    Therefore the proof is completed.
\end{proof}
\begin{remark}
From the proof we can see that the global Lipschitz condition Assumption \ref{assumption: lip in rkhs} can be relaxed to a Lipschitz condition within a subset of $U\subset\mathcal{H}$
\[U :=U_1\times\dots\times U_m,\]
where each $U_i$ is the convex hull of $\mu_i(F_i),\mu_i(\hat{F_i})$ and $\mu_i(\hat F_i^b)$.
\end{remark}

A specific function that satisfies our assumption is $$\phi(\mu_1(F_1),\cdots,\mu_m(F_m))=\sum_{i=1}^m\left\langle f_i,\mu_i(F_i) \right\rangle^2_{\mathcal{H}_i},$$
where $f_i:\mathbb R\to\mathbb R$ is a function that can be embeded in $\mathcal{H}_i$ and $ \|f_i^4\|_{\mathcal{H}_i}<\infty$ for all $i\in [m]$. We have
$$\partial_i\phi(\mu(F))=2\left\langle f_i,\mu_i(F_i) \right\rangle f_i,$$
which is nontrivial and Lipschitz continuous with Lipschitz constant $2\|f_i\|^2_{\mathcal{H}_i}$, so $D\phi$ is also Lipschitz continuous by Theorem \ref{thm: partial and total}. Moreover, $ \|f_i^4\|_{\mathcal{H}}<\infty$ implies $\|\partial_i\phi(\mu(F))^4\|_{\mathcal{H}_i}<\infty$. A direct consequence of this example is that $(\mathbb E X)^2$ can be proved to be simulated by our method efficiently.

Another example is the finite-horizon performance measure proposed in \cite{lam2022subsampling}. For convenience we restate the performance measure below. The performance measure is of the form
\begin{equation*}
    \phi(F_1,\dots,F_m)=\mathbb E_{F_1,\dots,F_m} h(\bX_1,\dots,\bX_m)
\end{equation*}
where $\bX_i=(X_i(1),\cdots,X_i(T_i))$ represents the $i$-th input process consisting of $T_i$ i.i.d. random variables distributed under $F_i$, each $T_i$ being a deterministic time, and $h$ is a performance function which satisfies Assumption 8 and Assumption 9 in \cite{lam2022subsampling}.
It can actually be regarded as a special form of our performance measure as $\phi$ can also be written as 
\[\phi = \langle h, \mu_1(F_1)\otimes\cdots\otimes \mu_m(F_m)\rangle_{\mathcal{H}_1\otimes\cdots\otimes\mathcal{H}_m}.\]
Noting that $D^{s}\phi = 0$ where $s=\sum_{i=1}^m T_i+1$, by Theorem \ref{thm: multi Taylor} we can see that the local Lipschitz condition is satisfied.

\section{Additional Experiments}\label{appendix:experiment}
\subsection{Debiasing}
In Figure \ref{fig:RMSE} we provide the root mean square error (RMSE) for the four debiasing tasks in Section \ref{subsection:biascorrection}. The precise definition of RMSE and BIAS is
\[\text{RMSE}=\sqrt{\sum_{i=1}^{1000}(\hat{\phi}-\phi)^2},\quad \text{BIAS}=\sum_{i=1}^{1000}|\hat{\phi}-\phi|.\]
We also provide the median (50\% percentile) of ten bootstrap resampling procedure in Table \ref{table: debias}.

\begin{table*}[]
\small
\centering
\caption{Debiasing performances with different bootstrap methods: Standard Bootstrap and Orthogonal Bootstrap.}
\label{table: debias}
\vspace{0.1in}
\begin{tabular}{cc||cc||cc||cc||cc}
\hline
\multirow{2}{*}{}    & \multirow{2}{*}{$B$} & \multicolumn{2}{c||}{{\textbf{Ellipsoidal}}}   & \multicolumn{2}{c||}{{\textbf{Polynomial}}}            & \multicolumn{2}{c||}{{\textbf{Entropy}}}       & \multicolumn{2}{c}{\textbf{Optimization}}               \\ \cline{3-10} 
                     &                     & \multicolumn{1}{c|}{\textbf{RMSE}} & {\textbf{BIAS}} & \multicolumn{1}{c|}{\textbf{RMSE}} & {\textbf{BIAS}}  & \multicolumn{1}{c|}{\textbf{RMSE}} & {\textbf{BIAS}}  & \multicolumn{1}{c|}{\textbf{RMSE}} & {\textbf{BIAS}} \\\hline\hline
Standard Bootstrap   &       2               & \multicolumn{1}{c|}{7.23}                  &    {{197.9}}                  & \multicolumn{1}{c|}{17.57}                  &     {{501.4}}     & \multicolumn{1}{c|}{0.893}                  &   {{24.23}}   & \multicolumn{1}{c|}{1.734}                  &   {{49.52}}             \\ \hline
Orthogonal Bootstrap &   2                   & \multicolumn{1}{c|}{6.63}                  &  {{168.9}}                     & \multicolumn{1}{c|}{14.92}                  &    {{395.8}}      & \multicolumn{1}{c|}{0.836}                  &   {{21.15}}   & \multicolumn{1}{c|}{1.598}                  &   {{45.44}}            \\ \hline\hline
Standard Bootstrap   &       5               & \multicolumn{1}{c|}{6.80}                  &          {{178.8}}            & \multicolumn{1}{c|}{15.69}                  &   {{428.7}}      & \multicolumn{1}{c|}{0.862}                  &   {{22.26}}   & \multicolumn{1}{c|}{1.404}                  &   {{38.57}}                 \\ \hline
Orthogonal Bootstrap &   5                   & \multicolumn{1}{c|}{6.61}                  &          {{167.0}}             & \multicolumn{1}{c|}{14.79}                  &           {{388.2}}   & \multicolumn{1}{c|}{0.833}                  &   {{20.92}}   & \multicolumn{1}{c|}{1.355}                  &   {{36.70}}   \\ \hline\hline
Standard Bootstrap   &       10               & \multicolumn{1}{c|}{6.68}                  &   {{172.3}}                  & \multicolumn{1}{c|}{14.99}                  &    {{404.8}}     & \multicolumn{1}{c|}{0.840}                  &   {{21.53}}   & \multicolumn{1}{c|}{1.281}                  &   {{33.96}}          \\ \hline
Orthogonal Bootstrap &   10                   & \multicolumn{1}{c|}{6.61}                  &    {{166.5}}                 & \multicolumn{1}{c|}{14.66}                  &       {{383.3}}     & \multicolumn{1}{c|}{0.833}                  &   {{20.92}}   & \multicolumn{1}{c|}{1.246}                  &   {{32.82}}       \\ \hline\hline
Naive Estimator &                      & \multicolumn{1}{c|}{10.30}                  &   {{271.1}}                  & \multicolumn{1}{c|}{25.84}                  &       {{634.7}}    & \multicolumn{1}{c|}{1.802}                  &   {{51.18}}   & \multicolumn{1}{c|}{4.973}                  &   {{153.8}}            \\ \hline

\end{tabular}
\end{table*}

\begin{figure}
    \centering
    \includegraphics[width=3in]{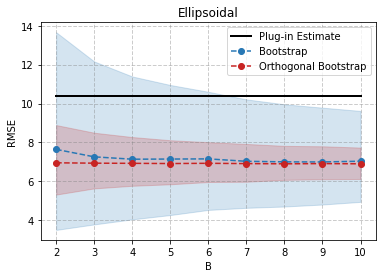}
    \includegraphics[width=3in]{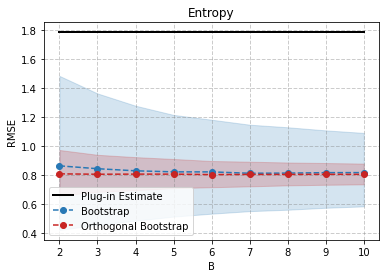}
    \includegraphics[width=3in]{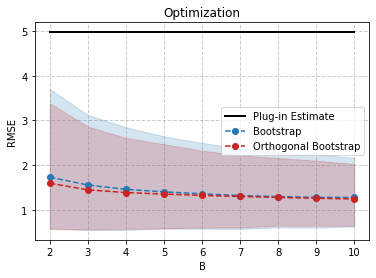}
    \includegraphics[width=3in]{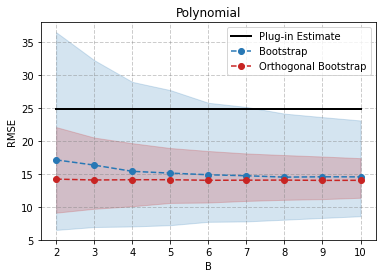}
    \caption{Orthogonal Bootstrap can significantly reduce the simulation output for the examples shown in \cite{ma2022correcting} when the number of Bootstrap resampling is limited. The $x$-axis represents the time of Bootstrap resampling and $y$-axis denotes the root mean square error produced by the estimation. The shaded area represents the 80\% quantile interval for repeated simulations. Orthogonal Bootstrap can significantly reduce the simulation variance.}
    \label{fig:RMSE}
\end{figure}

\subsubsection{Calculation of Influence Function for Constrained Optimization Problems }\label{appendix: constrained optimization}
In this section we provide a method of calculating the influence function for constrained optimization problem with randomness. 
\begin{figure}
    \centering
    \includegraphics[width=3in]{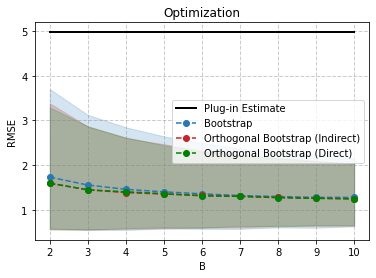}
    \includegraphics[width=3in]{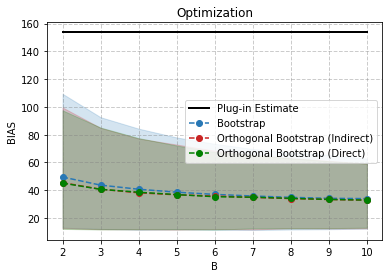}
    \caption{Comparison of indirect and direct influence function calculation in the constrained optimization problem.}
    \label{fig:vs}
\end{figure}

For a general constrained optimization problem with randomness
\begin{align*}
\text{Minimize} \quad & f(x,\xi) \\
\text{subject to} \quad & h_i(x,\xi) \leq 0, \quad i = 1, 2, \ldots, m \\
& g_j(x,\xi) = 0, \quad j = 1, 2, \ldots, p \\
\text{where} \quad & x \in \mathbb{R}^n\text{ and } \xi\in\mathbb R^k\text{ is a random vector},
\end{align*}

consider its Lagrangian
\begin{equation*}
    L(x,\mu,\nu)=f(x,\xi)+\sum_{i=1}^m\mu_ig_i(x,\xi)+\sum_{j=1}^p\nu_jh_j(x,\xi)
\end{equation*}
where $\mu_i$ are the dual variables on the equality constraints and $\nu_j\ge 0$ are the dual variables on the inequality constraints. Regard $\xi$ as a fixed parameter for now, the KKT conditions for stationarity, primal feasibility, and complementary slackness are
\begin{align*}
& \nabla_x f(x^*,\xi) + \sum_{i=1}^{m} \mu_i^* \nabla_x g_i(x^*,\xi) + \sum_{j=1}^{p} \nu_j^* \nabla_x h_j(x^*,\xi) = 0 \\
& h_j(x^*,\xi) = 0, \quad j = 1, 2, \ldots, p \\
& \mu^*_i g_i(x^*,\xi) = 0, \quad i = 1, 2, \ldots, m \\
\end{align*}
where $x^*$, $\mu_i^*$, and $\mu_j^*$ are the optimal primal and dual variables. Now, unfix $\xi$. Taking the differentials of these conditions yields the equation
\begin{align*}
&\left( \nabla_{xx} f(x^*,\xi)+\sum_{i=1}^{m} \mu_i^* \nabla_{xx} g_i(x^*,\xi)+\sum_{j=1}^{p} \nu_j^* \nabla_{xx} h_j(x^*)\right) \mathrm{d}x^* + \sum_{i=1}^{m} \nabla_{x} g_i(x^*)\mathrm{d}\mu_i^* + \sum_{j=1}^{p} \nabla_x h_j(x^*)\mathrm{d}\nu_j^* \\+& \left( \nabla_{x\xi} f(x^*,\xi)+\sum_{i=1}^{m} \mu_i^* \nabla_{x\xi} g_i(x^*,\xi)+\sum_{j=1}^{p} \nu_j^* \nabla_{x\xi} h_j(x^*)\right) \mathrm{d}\xi= 0 \\
& \nabla_x h_j(x^*,\xi)\mathrm{d} x^* +\nabla_\xi h_j(x^*,\xi)\mathrm{d} \xi = 0, \quad j = 1, 2, \ldots, p \\
& g_i(x^*,\xi)\mathrm{d}\mu_i^* + \mu^*_i \nabla_x g_i(x^*,\xi)\mathrm{d} x^* + \mu^*_i \nabla_\xi g_i(x^*,\xi)\mathrm{d} \xi = 0, \quad i = 1, 2, \ldots, m \\
\end{align*}
The optimal primal and dual variables can be calculated via numerical method. Therefore, we can solve the above equation to determine $\frac{\partial x^*}{\partial \xi}$, $\frac{\partial \mu^*}{\partial \xi}$, and $\frac{\partial \nu^*}{\partial \xi}$. If these quantities behave well, then we can combine them with the chain rule to determine the influence function with respect to the distribution of $\xi$.

The optimization problem we have selected for this study has a particularly tractable solution, making it amenable to differentiation and, consequently, to the computation of the influence function. In Figure \ref{fig:vs}, we present a comparative analysis of these two approaches for calculating the influence function. As anticipated, calculating the influence function via the explicit solution of the constrained optimization problem yields the same result in comparison to the implicit way of determining the influence function. Therefore, for optimization problem without a tractable solution, our implicit way of determining the influence function can be powerful. 

\subsection{Confidence Interval Construction}
The confidence interval constructed by the bootstrap method is

\begin{equation*}
    \begin{aligned}
        [\hat\phi-z_{1-\alpha/2}\sqrt{\Variance\left(\phi(\hat F_{1}^{b},\cdots,\hat F_{m}^{b})\right)}, \hat\phi+z_{1-\alpha/2}\sqrt{\Variance\left(\phi(\hat F_{1}^{b},\cdots,\hat F_{m}^{b})\right)}],
    \end{aligned}
\end{equation*}

where $z_{1-\alpha/2}$ being the $(1-\alpha/2)$-quantile of the standard normal, $\hat\phi=\phi(\hat F_1,\cdots,\hat F_m)$ is the plug-in estimator of $\phi(F_1,\cdots,F_m)$ and $\hat F_i=\frac{1}{n}\sum_{j=1}^n \delta_{X_{i,j}}$. Therefore using Orthogonal Bootstrap to simulate the variance, we arrive at Algorithm \ref{alg:ciconstructionob} for constructing confidence interval.
\begin{algorithm*}
\caption{Confidence Interval Construction via Orthogonal Bootstrap}\label{alg:ciconstructionob}
 \textbf{Input}: A generic performance measure $\phi(F_1,\cdots,F_m)$, i.i.d samples  $\{X_{i,1},\cdots,X_{i,n_i}\}\in\mathbb{R}^{d_1}$ of $F_i$, influence function $\mathcal{I}_i^\phi$ of $\phi$ respect to $\hat{F}_i$, and confidence level $\alpha$.  \\
 \textbf{Output}: $(1-\alpha)$ prediction interval of {$\phi(F_1,\cdots,F_m)$}.
\begin{algorithmic}
\STATE $\hat \phi\leftarrow \phi(\hat F_1,\cdots,\hat F_m)$, where $\hat F_i=\frac{1}{n_i}\sum_{j=1}^{n_i}\delta_{X_{i,j}}$

\FOR{b=1:B}
\FOR{i=1:m}
\STATE Sample $\{x_{i,1}^b,\cdots,x_{i,n_i}^b\}$ i.i.d from $\hat{F_i}$
\ENDFOR
\STATE $\hat\phi^b\leftarrow\phi(\hat F_1^b,\cdots,\hat F_m^b)$, where $\hat{F}_i^b=\frac{1}{n_i}\sum_{j=1}^{n_i} \delta_{x_{i,j}^b}$
\STATE Calculate $\mathcal{\hat I}^b=\sum_{i=1}^m\frac{1}{n_i}\sum_{j=1}^{n_i}\mathcal{I}_i^\phi(\tilde{x}_{i,j})$
\ENDFOR
\STATE Calculate $\overline{\phi-\mathcal{I}}\leftarrow\frac{1}{B}\sum_{b=1}^{B}(\hat\phi^b-\mathcal{\hat I}^b)$, $\overline{\mathcal{I}}\leftarrow\frac{1}{B}\sum_{b=1}^{B}\mathcal{\hat I}^b$.
\STATE Construct the $1-\alpha$-prediction interval as $[\hat\phi-z_{1-\alpha/2}S,\hat\phi+z_{1-\alpha/2}S],$ where 
\begin{equation}
    \begin{aligned}
        S^2&=\sum_{i=1}^m\frac{1}{n_i^2}\sum_{j=1}^{n_i}(\mathcal{I}_i^\phi (X_{i,j}))^2+\frac{1}{B}\sum_{b=1}^{B}\left(\hat\phi^b-\mathcal{\hat I}^b-\overline{\phi-\mathcal{I}}\right)^2+\frac{2}{B}\sum_{b=1}^{B}\left(\hat\phi^b-\mathcal{\hat I}^b-\overline{\phi-\mathcal{I}}\right)\left(\mathcal{\hat I}^b-\overline{\mathcal{I}}\right),
    \end{aligned}
\end{equation}
and $z_{1-\alpha/2}$ is the $(1-\alpha/2)$-quantile of the standard normal. 
\end{algorithmic}
\end{algorithm*}

\subsection{Real Data}
\subsubsection{Prediction Interval Construction via Orthogonal Bootstrap}
Suppose that the input data and target data, $\bX$ and $\by$, are observed, and a neural network is trained on these data to produce an estimated regression function $\hat{f}$. Then $B$  Monte Carlo replications are collected, and a neural network is trained on each to produce regression functions $f^b$, $b = 1,\dots, B$. Now, suppose a new feature vector, $\bx_{test}$, is observed and it is desired to provide a point estimate and a prediction interval (PI) for its unknown target value, $y_{test}$.

Prediction interval is slightly different from confidence interval because the model possess inherent irreducible error \cite{contarino2022constructing,khosravi2011review}. To construct prediction interval, besides the term $S^2$ in confidence interval construction, we need to add an additional term $\sigma^2$ for the irreducible error of the model. Specifically, we follow the pivot bootstrap method described in \cite{contarino2022constructing}.
The irreducible error $\sigma^2$ can be estimated from the residuals of the out-of-sample predictions. For a bootstrap resample $\{\bX_i^b\}_{i=1}^n$, the corresponding set of out-of-bag observations is $\{\bX_i|\bX_i\not\in\{\bX_i^b\}_{i=1}^n\}$ and is denoted here as $\bX^b_{oob}$. Then, for each bootstrap resample, $\sigma^2_b$ is:
\[\sigma^2_b=\frac{\sum_{\bx\in \bX^b_{oob}} (y-f^b(\bx))^2}{n^b}\]
where $n^b$ is the number of out-of-bag resamples. We summarize our algorithm in Algorithm \ref{alg:piconstructionob}.

\begin{algorithm*}
\caption{Prediction Interval Construction via Orthogonal Bootstrap}\label{alg:piconstructionob}
 \textbf{Input}: Number of training data $n$, training data $\bX$ and $\by$, test observation $\bx_{test}$, learning algorithm $L$, desired number of  Monte Carlo replications $B$, and desired coverage probability $1 - \alpha$.  \\
 \textbf{Output}: $(1-\alpha)$ prediction interval of {$y_{test}$}.
\begin{algorithmic}
\STATE Train learning algorithm $L$ on $\bX$ and $\by$; denote the trained regressor as $\hat{f}$
\STATE $\hat{\phi}\leftarrow \hat{f}(\bx_{test})$
\STATE For every input data $\bX_i\in\bX$, calculate the influence function $\mathcal{I}(\bX_i)$ of $\hat{\phi}$, 
\FOR{b=1:B}
\STATE Generate bootstrap sample $(\bX^b,\by^b)$ from $(\bX,\by)$
\STATE Determine the out of bag samples $(\bX^b_{oob},\by^b_{oob})$
\STATE Train learning algorithm $L$ on $\bX^b$ and $\by^b$; denote the trained regressor as $f^b$
\STATE $\hat\phi^b\leftarrow f^b(\bx_{test})$
\STATE $\mathcal{\hat I}^b=\frac{1}{n}\sum_{i=1}^{n}\mathcal{I}(\bX^b_i)$
\ENDFOR
\STATE $\overline{\phi-\mathcal{I}}\leftarrow\frac{1}{B}\sum_{b=1}^{B}(\hat\phi^b-\mathcal{\hat I}^b)$, $\overline{\mathcal{I}}\leftarrow\frac{1}{B}\sum_{b=1}^{B}\mathcal{\hat I}^b$.
\STATE Construct the $1-\alpha$-confidence interval as $$[\hat\phi-z_{1-\alpha/2}\sqrt{S^2+\sigma^2},\hat\phi+z_{1-\alpha/2}\sqrt{S^2+\sigma^2}],$$ where 
\begin{equation}
    S^2=\left\{\begin{matrix}
  S_1^2& \quad\text{if}\quad S_1^2\geq0\\
  S_2^2& \quad\text{if}\quad S_1^2<0
\end{matrix}\right.,
\end{equation}

\begin{equation}
    \begin{aligned}
        S_1^2&=\frac{1}{n^2}\sum_{i=1}^n(\mathcal{I} (\bX_{i}))^2+\frac{1}{B}\sum_{b=1}^{B}\left(\hat\phi^b-\mathcal{\hat I}^b-\overline{\phi-\mathcal{I}}\right)^2+\frac{2}{B}\sum_{b=1}^{B}\left(\hat\phi^b-\mathcal{\hat I}^b-\overline{\phi-\mathcal{I}}\right)\left(\mathcal{\hat I}^b-\overline{\mathcal{I}}\right),
    \end{aligned}
\end{equation}
\begin{equation}
    \begin{aligned}
        S_2^2&=\frac{1}{n^2}\sum_{i=1}^n(\mathcal{I} (\bX_{i}))^2,
    \end{aligned}
\end{equation}
and $z_{1-\alpha/2}$ is the $(1-\alpha/2)$-quantile of the standard normal. 
\end{algorithmic}
\end{algorithm*}

\subsubsection{Details}\label{appendix:pi experiment detail}
In this section we provide the details of the real data experiments in section.
For all real data examples, we employ the Adam optimizer with default hyperparameters, with the exception of setting the weight decay to 0.01. The training loss is set to be the squared loss, i.e. $L(\bx,\theta)=(f_\theta(\bx)-y)^2$, where $f_\theta(\bx)$ is parameterized by a two-layer neural network with hidden dimension 100 .

For the Yacht dataset, we utilize a batch size of 64 and train for 500 epochs. We use 80\% data for training and 20\% for testing. 
For the Energy dataset, we utilize a batch size of 128 and train for 250 epochs. We use 70\% data for training and 30\% for testing. 
For the kin8nm dataset, we utilize a batch size of 256 and train for 150 epochs. We use 95\% data for training and 5\% for testing.
The specific choice of hyperparameters serves the dual purpose of ensuring that the neural networks fit the data effectively while also ensuring that the inherent variance of the model remains within a reasonable range compared to the bootstrapped variance.

The influence function of the model parameters $\theta$ is
\begin{equation}
    \begin{aligned}
        \mathcal{I}^{\theta}(\bx)
        &=-H_{\theta}^{-1}\nabla_\theta L(\bx, \theta),
    \end{aligned}
\end{equation}
therefore by the chain rule of derivatives, the influence function at a particular test point $\bx_\text{test}$ is 
\begin{equation}
    \begin{aligned}
        \mathcal{I}^{\bx_\text{test}}(\bx)
        &=-\nabla_\theta f_\theta (\bx_\text{test})^TH_{\theta}^{-1}\nabla_\theta L(\bx, \theta).
    \end{aligned}
\end{equation}

We use the conjugate gradient method to calculate the influence function \cite{koh2017understanding} \cite{martens2010hessianfree}. 

\end{document}